\documentclass[lettersize,journal]{IEEEtran}

\usepackage{cite}
\usepackage{amsmath,amssymb,amsfonts}
\usepackage{graphicx,color}
\usepackage{textcomp}
\usepackage{makecell}
\usepackage{algorithm}
\usepackage{array}
\usepackage[caption=false,font=normalsize,labelfont=sf,textfont=sf]{subfig}
\usepackage{multirow}
\usepackage{booktabs}
\usepackage{stfloats}
\usepackage{url}
\usepackage{verbatim}
\usepackage{graphicx}
\usepackage{cite}
\usepackage{comment}
\usepackage{amsthm}

\usepackage{enumitem}
\usepackage{xcolor}
\usepackage{algorithm}
\usepackage{algpseudocode}
\usepackage{amsmath}
\usepackage{booktabs}
\usepackage{stmaryrd} 

\DeclareMathOperator*{\argmax}{arg\,max}

\def\BibTeX{{\rm B\kern-.05em{\sc i\kern-.025em b}\kern-.08em
    T\kern-.1667em\lower.7ex\hbox{E}\kern-.125emX}}
\usepackage{balance}
\newcommand{\addb}[1]{{\color{black}{#1}}}

\newcommand{\kripke}{\mathcal{M}_t}
\newcommand{\everyK}{\mathcal{E}}
\newcommand{\everyA}{\mathcal{A}}
\newcommand{\timeD}{T}
\newcommand{\timeV}{t_v}

\newcommand{\trecE}{t_{\mathrm{rec}}^{\mathrm{epi}}}
\newcommand{\tdurE}{t_{\mathrm{dur}}^{\mathrm{epi}}}
\newcommand{\trecA}{t_{\mathrm{rec}}^{\mathrm{act}}}
\newcommand{\tdurA}{t_{\mathrm{dur}}^{\mathrm{act}}}

\newcommand{\agentSet}{N}
\newcommand{\agentSetdet}{N_\mathrm{det}}
\newcommand{\graph}{G}
\newcommand{\edges}{E}
\newcommand{\agentsize}{n}
\newcommand{\neighbors}{V}
\newcommand{\Xspace}{X}
\newcommand{\Xstate}{\boldsymbol{x}}
\newcommand{\Yspace}{Y}
\newcommand{\Mspace}{M}
\newcommand{\msg}{m}
\newcommand{\Ystate}{\boldsymbol{y}}
\newcommand{\Ystateh}{\hat{\boldsymbol{y}}}
\newcommand{\Ispace}{I}
\newcommand{\Istate}{\iota}
\newcommand{\uExt}{\boldsymbol{u}^{\mathrm{ext}}}
\newcommand{\UExt}{U^{\mathrm{ext}}}

\newcommand{\UExti}{U^{\mathrm{ext}}_i}

\newcommand{\uExtit}{u^{\mathrm{ext}}_{i,t}}

\newcommand{\uExto}{u^{\mathrm{ext}}_1}
\newcommand{\uExtt}{u^{\mathrm{ext}}_2}
\newcommand{\uExtn}{u^{\mathrm{ext}}_n}

\newcommand{\UEps}{U^{\mathrm{epi}}}
\newcommand{\uEpsit}{u^{\mathrm{epi}}_{i,t}}

\newcommand{\extmap}{f}
\newcommand{\senmap}{s}
\newcommand{\worlds}{W}

\newcommand{\val}{v}

\newcommand{\pSet}{\Omega}
\newcommand{\p}{p}
\newcommand{\sco}{\lambda}
\newcommand{\scos}{\Lambda}
\newcommand{\worldw}{w}
\newcommand{\worldv}{w'}
\newcommand{\worldwt}{w_t}

\newcommand{\rel}{R}
\newcommand{\run}{r}
\newcommand{\m}{r}

\newcommand{\etaE}{\eta^{\mathrm{epi}}}
\newcommand{\etaA}{\eta^{\mathrm{act}}}

\newcommand{\taumax}{\tau_{\mathrm{max}}}
\newcommand{\epPi}{\pi_i^{\text{epi}}}

\newcommand{\drate}{\delta}
\newcommand{\knows}{\mathcal{K}}
\newcommand{\possible}{\mathcal{P}}
\newcommand{\glob}{\mathcal{G}}
\newcommand{\until}{\mathcal{U}}
\newcommand{\lang}{\mathcal{L}}
\newcommand{\formula}{\varphi}

\newcommand{\Dspace}{D}
\newcommand{\dst}{d}

\newcommand{\conf}{\zeta}
\newcommand{\arm}{a}
\newcommand{\var}{\sigma}
\newcommand{\dstd}{d_o}
\newcommand{\tvio}{t_v}
\newcommand{\tdet}{t_\text{det}}
\newcommand{\ndct}{n_{\text{steps}}}

\newcommand{\lcomm}{l_{\text{comm}}}
\newcommand{\link}{l}

\newcommand{\mean}{\mu}

\newcommand{\maxRE}{\alpha_1}
\newcommand{\minDE}{\beta_1}
\newcommand{\maxRA}{\alpha_2}
\newcommand{\minDA}{\beta_2}

\newcommand{\kpros}{K}
\newcommand{\kpro}{k}
\newcommand{\lpro}{l}

\newcommand{\resS}{\mathcal{R}_{\mathrm{system}}}
\newcommand{\resE}{\mathcal{R}_{\mathrm{epistemic}}}
\newcommand{\resA}{\mathcal{R}_{\mathrm{action}}}
\newtheorem{definition}{Definition}

\newtheorem{remark}{Remark}
\newtheorem{assumption}{Assumption}
\newtheorem{example}{Example}
\newtheorem{theorem}{Theorem}

\newtheorem{corollary}{Corollary}
\def\BibTeX{{\rm B\kern-.05em{\sc i\kern-.025em b}\kern-.08em
    T\kern-.1667em\lower.7ex\hbox{E}\kern-.125emX}}
\AtBeginDocument{\definecolor{ojcolor}{cmyk}{0.93,0.59,0.15,0.02}}

\begin{document}

\title{Logic-Driven Semantic Communication for Resilient Multi-Agent Systems}

\author{Tamara Alshammari, \textit{Student Member}, \textit{IEEE} and Mehdi Bennis, \textit{Fellow}, \textit{IEEE}
\thanks{All authors are from Centre for Wireless Communication, University of Oulu, Finland (emails: \{firstname.lastname\}@oulu.fi).}}

\maketitle

\begin{abstract}
The advent of 6G networks is accelerating autonomy, intelligence, and interconnectedness in large-scale, decentralized multi-agent systems (MAS). While this evolution enables highly adaptive behavior, it also heightens vulnerability to stressors such as environmental changes, faults, misinformation, and adversarial behavior. Existing literature on resilience in decentralized and multi-agent systems largely focuses on isolated aspects, such as fault tolerance or robustness, without offering a principled, unified definition of multi-agent resilience. This gap limits the ability to design systems that can continuously sense, adapt, and recover under dynamic conditions. This article proposes a formal definition of MAS resilience grounded in two complementary dimensions: epistemic resilience, wherein agents recover and sustain accurate knowledge of the environment, and action resilience, wherein agents leverage that knowledge to coordinate and sustain goals under disruptions. We formalize resilience via temporal epistemic logic and quantify it using recoverability time (how quickly desired properties are re-established after a disturbance) and durability time (how long accurate beliefs and goal-directed behavior are sustained after recovery). We design an agent architecture and develop decentralized algorithms to achieve both epistemic and action resilience. We further provide formal verification guarantees, showing that our specifications are sound with respect to the metric bounds and admit finite-horizon verification, enabling design-time certification and lightweight runtime monitoring. Through a case study on distributed multi-agent decision-making under abrupt network stressors, we show that our approach consistently outperforms baseline methods. Collectively, our formal verification analysis and simulation results highlight that the proposed framework enables resilient, knowledge-driven decision-making and sustained operation, laying the groundwork for resilient decentralized multi-agent systems in next-generation communication infrastructures.
\end{abstract}


\section{Introduction}
The advent of 6G envisions pervasive autonomy, distributed intelligence, and dense interconnection among devices, users, and infrastructure, forming large-scale, decentralized multi-agent systems (MAS) capable of collective sensing, reasoning, and adaptation. Yet this very complexity heightens exposure to stressors - ranging from environmental shifts and component faults to misinformation cascades and adversarial behavior - making \textbf{resilience} a first-class performance metric for future networks\cite{khaloopour2024resilience,vesterby_2022,9963527,weissberger2023imt2030,cioschina2023miit}. 

Achieving such resilience, however, cannot rely on traditional centralized architectures. As 6G systems evolve toward massive scale and heterogeneity, reliance on centralized control becomes increasingly impractical. Ultra-dense device deployments, intermittent connectivity, stringent latency constraints, and privacy or administrative boundaries all limit the feasibility of global orchestration~\cite{bennis_urllc}. Centralized architectures also introduce single points of failure and are inherently vulnerable to communication delays, congestion, and attacks~\cite{AI_dec}\cite{dec1}. In contrast, \emph{decentralized MAS} distribute intelligence and decision-making across agents that sense, reason, and act locally while coordinating through limited communication~\cite{MAS1}\cite{MAS2}. Such decentralization enables faster adaptation to local disturbances. Consequently, any rigorous notion of resilience for future generations of wireless networks must be grounded in decentralized operation, where agents maintain functionality and alignment despite partial observability and dynamic network conditions.

\textbf{Motivation}. This challenge extends beyond wireless systems, reflecting a broader, multidisciplinary difficulty in defining resilience rigorously~\cite{Royce,mitre,scheffer2009early,dai2012generic,barker2013resilience,fang2016resilience,liao2018resilience,hulse2022understanding,wu2024towards}.
In wireless systems, resilience has become a central concern for next-generation (6G) systems~\cite{khaloopour2024resilience,resil,brothersarms}, a shift underscored by initiatives such as the NSF Resilient \& Intelligent NextG Systems (RINGS I–II) programs~\cite{WinNT} and by the recognition that communication infrastructures are now critical societal utilities exposed to both natural hazards and human-induced disruptions~\cite{brothersarms}. \addb{This need for fast, distributed resilience arises prominently in several
emerging 6G scenarios. For example, in resilient vehicle platooning, unexpected changes in road conditions, sudden interference, or sensing faults
can momentarily degrade coordination performance, yet safe operation demands a
rapid return to nominal behavior. Likewise, distributed sensor networks deployed for
disaster response may experience sensor failures, signal blockage, or
environmental shocks that temporarily reduce sensing quality, requiring the system
to reorganize and restore effectiveness. Similarly, UAV swarms and robotic
teams can encounter dynamic obstacles, shifting propagation environments, or
unexpected load changes that momentarily impair their operational efficiency,
necessitating fast and sustained recovery.} However, resilience in networked systems still lacks a principled, unified definition, leading to broad interpretations and uneven engineering practices~\cite{khaloopour2024resilience,vesterby_2022,9963527,weissberger2023imt2030,cioschina2023miit}. 
Resilience is often conflated with robustness or reliability, necessitating their distinction. Reliability, prominent in 5G, quantifies the probability of meeting performance targets and focuses on tail-risk statistics for metrics like latency and rate~\cite{8472907,rel1}. Robustness approaches typically rely on predefined uncertainty models and aim to optimize performance under adverse conditions. This is often formalized through frameworks such as min–max optimization, stochastic modeling, or risk-sensitive games, where the system is designed to withstand worst-case or probabilistic scenarios~\cite{ieeespectrum,rel2}. Both are indispensable, but they presume stressors that are sufficiently specified a priori. By contrast, resilience posits that unexpected events will occur and emphasizes the capacity to detect, localize, adapt online so that the system sustains its goals despite model violations and shifting conditions~\cite{ieeespectrum,resil}. The persistence of this gap in resilience definition stems from the immaturity of the underlying mathematical foundations required to rigorously formalize and quantify resilience in networked systems, thereby limiting both design-time analysis and runtime guarantees~\cite{ieeespectrum,resil,brothersarms}.

\textbf{Challenges and Contributions}. Bridging this gap requires not only formal definitions but also operational principles for how multi-agent systems adapt and maintain functionality under uncertainty. Indeed, coordinated response to unknown stressors hinges on continuously updating a \emph{shared situational picture} where agents sense, communicate, and cross-check as conditions evolve, so that decisions are guided by a living internal model of the world rather than static plans~\cite{prorok2021beyond,resil}. In this work, \textbf{we propose a new system-wide definition of resilience in MAS} defined along two coupled dimensions: (i) \emph{epistemic resilience} which refers to the agents’ ability to detect knowledge gaps, take epistemic actions to address those gaps, recover from discrepancies between their internal models (beliefs) about the environment and its actual state, and sustain accurate internal models after recovery and (ii) \emph{action resilience} which concerns the agents’ capacity to rapidly adapt their external action policies in light of updated knowledge, by implementing recovery protocols that re-establish effective alignment between current beliefs and physical actions, and sustain optimal external action policies after recovery. To operationalize this definition, we introduce metrics that attach directly to the epistemic and action loops: \emph{epistemic recoverability time} (how quickly the MAS regains accurate shared/consistent internal model after a disturbance), \emph{epistemic durability time} (how long the MAS sustains accurate internal models after epistemic recovery), \emph{action recoverability time} (how quickly the MAS regains optimal external action policies after epistemic recovery), and \emph{action durability time} (how long the MAS sustains these optimal policies after action recovery). To the best of our knowledge, this two-fold formal definition of resilience in networked systems with quantifiable system-level metrics is absent from prior 6G/MAS resilience literature, which treats isolated mechanisms or offer broad, non-unified notions of resilience~\cite{khaloopour2024resilience,9963527,ieeespectrum,resil}. 

To formally define this system-wide resilience and make it \emph{quantifiable} in MAS, we need a formal language that (i) represents what each agent knows, does not know, and consider possible about the external world; (ii) reasons about mutual knowledge that captures \emph{what all agents know} as a notion of group alignment, (iii) captures how knowledge/belief~\footnote{In this paper, the terms \emph{knowledge}, \emph{belief}, and \emph{epistemic} are used interchangeably.} evolves under sensing, communication, and stressors, and (iv) specifies the execution of optimal action policies at both the individual agent level and the group level. In fact, \emph{temporal epistemic logic} provides operators for individual and mutual knowledge with rigorous semantics and update mechanisms for modeling information evolution~\cite{TempoEpsLogic1}\cite{Epistemic1}. Crucially, these semantics also enable \emph{formal verification} which captures the ability to check, at design time or runtime, whether a multi-agent system satisfies resilience properties under specified assumptions. This capability is essential for certifying that recovery protocols, knowledge-update rules, and coordination strategies meet resilience requirements before deployment, as well as for runtime monitoring in dynamic environments.

Concretely, we instantiate temporal epistemic logic with \emph{Kripke structures}, which provide the semantic backbone for representing and evolving multi-agent knowledge~\cite{FHMV1995,Kripke1963SC,modal1,modal2,modal3}. Intuitively, a Kripke model lists all possible worlds, which are plausible configurations of the environment. It annotates which facts hold in each world and specifies, for each agent, which worlds are indistinguishable based on the agent’s available information. An agent ``knows'' a fact if that fact holds in all worlds the agent considers possible; mutual knowledge arises when all agents are aligned on a fact. As the system operates under sensing, communication, failures, or attacks, the model itself updates: accurate observations or informative messages eliminate false possibilities (tightening beliefs), whereas misinformation, sensor faults, or adversarial actions can reintroduce ambiguity or contradictions (inaccurate beliefs). Sequencing these models over time yields a natural account of how knowledge changes as the system operates, which is precisely what temporal epistemic logic formalizes~\cite{FHMV1995,vanDitmarsch2007,TempoEpsLogic1,Epistemic1}. Finally, this Kripke-based perspective connects directly to \emph{semantic communication} in AI-native 6G: a message is valuable when it reduces the receiver’s epistemic ambiguity~\cite{bennis_vision,Seo_semcom,AInative,Yang_semcom,semcom1,semcom2,semcom3,semcom4}. This view aligns naturally with resilience because reducing ambiguity accelerates epistemic recovery and stabilizes shared situational awareness, reducing action recovery and improving durability. In other words, communication strategies that prioritize uncertainty reduction in task-relevant knowledge directly enhance epistemic and action resilience, by enabling faster, more coordinated adaptation under stressors.

To this end, we outline our main contributions as follows:
\begin{itemize}
    \item We present a \textbf{formal definition of multi-agent resilience} and introduce quantifiable metrics grounded in two complementary dimensions: \emph{epistemic-based} and \emph{action-based} (\textbf{Section~\ref{sec:metrics}}).
    \item We \textbf{formalize} these metrics within a \textbf{temporal epistemic logic} framework, and \textbf{quantify} them through two parameters: \textbf{recoverability time} and \textbf{durability time}, which capture how quickly a system regains desired properties after disturbances and how long it sustains accurate beliefs and optimal performance (\textbf{Section~\ref{sec:metrics}}).
    \item To operationalize these metrics, we design a \textbf{semantic communication-enabled agent architecture} with coupled internal and external components. The internal component models each agent’s beliefs using Kripke structures, while enabling belief updates and alignment through logical message exchange and direct observations (\textbf{Section~\ref{sec:system_model}}). 
    \item We develop \textbf{decentralized algorithms} that achieve both epistemic and action resilience in decentralized multi-agent systems (\textbf{Section~\ref{sec:algorithms}}).
    \item We provide \textbf{formal verification guarantees}, proving that our resilience specifications are \textbf{sound} with respect to the defined metric bounds and admit \textbf{finite-horizon verification}. This enables both \textit{design-time certification} and \textit{lightweight runtime monitoring} following disturbances (\textbf{Theorem~\ref{thm:soundness}} and \textbf{Corollary~\ref{cor:completeness}}).
    \item Finally, we \textbf{demonstrate the effectiveness} of our framework through a distributed decision-making under abrupt network stressors case study, showing that it \textbf{outperforms existing baseline approaches} (\textbf{Section~\ref{sec:numerical}}).
\end{itemize}


\noindent \textbf{Paper Organization.}\quad
The remainder of this paper is organized as follows. Section~\ref{sec:system_model} presents the system model and the proposed agent architecture. Section~\ref{sec:kripke} introduces the temporal epistemic logic framework and Kripke structures. Section~\ref{sec:metrics} defines the proposed resilience metric, including its epistemic and action dimensions and their quantitative characterization in terms of recoverability and durability.
Section~\ref{sec:algorithms} details the decentralized algorithms designed to achieve epistemic and action resilience.
Section~\ref{sec:analytical} provides formal verification analysis of the proposed metric.
Section~\ref{sec:numerical} introduces and discusses the numerical simulations and performance evaluation.
Finally, Section~\ref{sec:conclusion} concludes the paper. 

\section{System Model and Agent Architecture}
\label{sec:system_model}
\begin{table*}[!t]
  \caption{Notation used in the system model, temporal–epistemic framework, and resilience metrics (Sections~\ref{sec:system_model} - \ref{sec:metrics}).}
  \label{tab:notation}
  \centering
  \footnotesize
  \renewcommand{\arraystretch}{1.15}
  \begin{tabular*}{\textwidth}{@{\extracolsep{\fill}} l p{0.80\textwidth}}
    \hline
    \textbf{Symbol} & \textbf{Description} \\
    \hline
    $\agentSet$ & Set of agents, with $|\agentSet|=\agentsize$. \\
    $\graph=(\agentSet,\edges)$ & Undirected communication graph. \\
    $\edges$ & Set of (undirected) communication links. \\
    $\neighbors_i$ & Neighbor set of agent $i$: $\{\,j\in\agentSet:\; ij\in\edges,\ j\neq i\,\}$. \\
    $\Xspace$ & External (physical) state space; $\Xstate_t\in\Xspace$ at time $t$. \\
    $\Dspace$ & Disturbance space; $\dst_t\in\Dspace$ at time $t$; $\dstd$ default (nominal). \\
    $\UExti$ & External action space of agent $i$; $\UExt=\prod_{i=1}^{\agentsize}\UExti$. \\
    $\uExt=(\uExto,\uExtt,\ldots,\uExtn)$ & Joint external action; $\uExtit$ is agent $i$’s action at time $t$. \\
    $\extmap:\Xspace\times\UExt\times\Dspace\to\Xspace$ & External transition map. \\
    $\Yspace$ & Observation space; $\senmap_i:\Xspace\to\Yspace$; $\Ystate_{i,t}=\senmap_i(\Xstate_t)$. \\
    $\Mspace_i$ & Set of (finite) message sequences agent $i$ can receive; $\msg_{j,t}$ from neighbor $j$. \\
    $\Ispace_i$ & Internal (epistemic) state space; $\Istate_{i,t}\in\Ispace_i$. \\
    $\Ystateh_{i,t+1}$ & Predicted observation for time $t{+}1$. \\
    $\UEps$ & Set of epistemic actions (\textsc{refine}, \textsc{revise}, \textsc{explore}, \textsc{broadcast}, \textsc{hold}); $u^{\mathrm{epi}}_{i,t+1}\in\UEps$. \\
    $\varrho_i$ & Prediction map. \\
    $\epPi$ & Epistemic policy.\\
    $\psi_i$ & Internal transition. \\
    $\pi_i^{\mathrm{ext}}$ & External policy. \\
    $\timeD$ & Discrete time domain $\{0,1,2,\ldots\}$. \\
    $\pSet$ & Set of atomic propositions; $\lang_{\agentsize}^{\text{time}}(\pSet)$ temporal–epistemic language. \\
    $\glob_{[0,\beta)}$, $\until_{[0,\alpha]}$ & Temporal operators “globally” and “until”. \\
    $\knows_i$, $\possible_i$ & Knowledge and possibility operators. \\
    $\everyK_\agentSet \formula$ & Mutual knowledge over the full agent set $\agentSet$ (used in resilience formulas). \\
    $\kripke$ & Kripke structure $(\worlds,\{\rel_{i,t}\}_{i\in\agentSet},\val)$. \\
    $\run$ & Run $\run:\timeD\to\worlds$; $\worldwt$ the actual world at time $t$. \\
    $\resS$ & System-wide resilience specification, defined as $\resS \equiv \resE \land \resA$. \\
    $\resE$ & Epistemic resilience specification; see \eqref{resE}. \\
    $\resA$ & Action resilience specification; see \eqref{resA}. \\
    $\tvio$ & Violation time at which the environmental change (stressor) occurs. \\
    $\trecE$ & Epistemic recovery time: The first time step when the agent epistemically recovered. \\
    $\Delta \trecE$ & Epistemic recoverability interval: The time needed to epistemically recover. \\
    $\tdurE$ & End time of epistemic durability: The last time step of having accurate beliefs. \\
    $\Delta \tdurE$ & Epistemic durability interval: The time interval in which the agent has accurate beliefs. \\
    $\trecA$ & Action recovery time. \\
    $\Delta \trecA$ & Action recoverability interval The first time step when the agent recovered its policies. \\
    $\tdurA$ & End time of action durability: The last time step of executing optimal policies.  \\
    $\Delta \tdurA$ & Action durability interval: The time interval in which the agent executes optimal policies.\\
    $\maxRE$ & Maximum allowable epistemic recovery time bound (used in \eqref{resE}). \\
    $\minDE$ & Minimum required epistemic durability duration (used in \eqref{resE}). \\
    $\maxRA$ & Maximum allowable action recovery time bound (used in \eqref{resA}). \\
    $\minDA$ & Minimum required action durability duration (used in \eqref{resA}). \\
    $\pi_i^{\text{opt}}$ & Atomic proposition: agent $i$ is currently acting optimally under the current environment;
    $\pi^{\text{opt}} \equiv \bigwedge_{i\in\agentSet}\pi_i^{\text{opt}}$. \\
    \hline
  \end{tabular*}
\end{table*}
\textbf{Communication Model}. Consider a system of $\agentsize$ agents forming a communication graph, represented as an undirected graph $\graph = (\agentSet, \edges)$, where $\agentSet$ is the set of agents of cardinality $\agentsize$, and $\edges$ is the set of edges representing communication links between the agents. An edge between agents $i$ and $j$ in $\agentSet$ is denoted by the unordered pair $ij = ji \in \edges$. The neighbors of agent $i$ are represented by the set $\neighbors_i = \{j \in \agentSet : ij \in \edges \text{ and } j \ne i\}$. 
\begin{figure*}[ht]
  \centering
  \subfloat[Proposed agent architecture.]{%
    \includegraphics[width=0.9\linewidth]{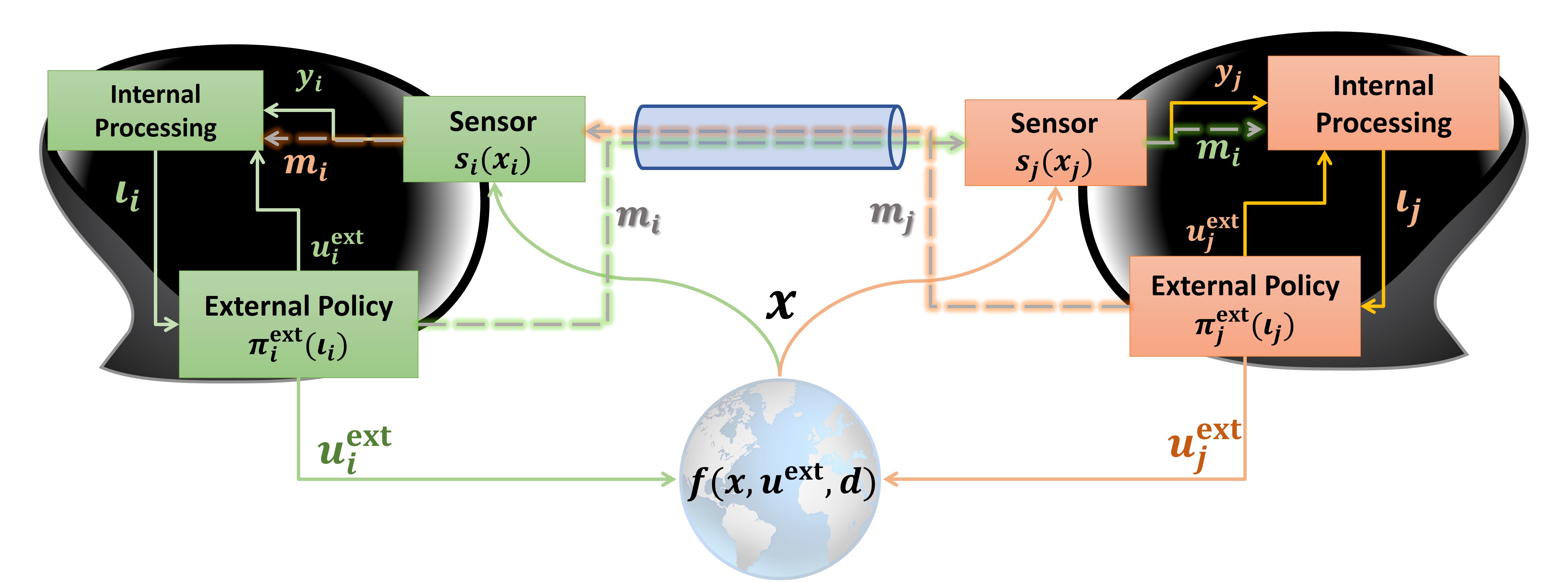}\label{fig:arch}}
  
  \par\medskip
  
  \subfloat[Detailed structure of agent~$i$'s `Internal Processing' unit within the proposed agent architecture.]{%
    \includegraphics[width=.8\linewidth]{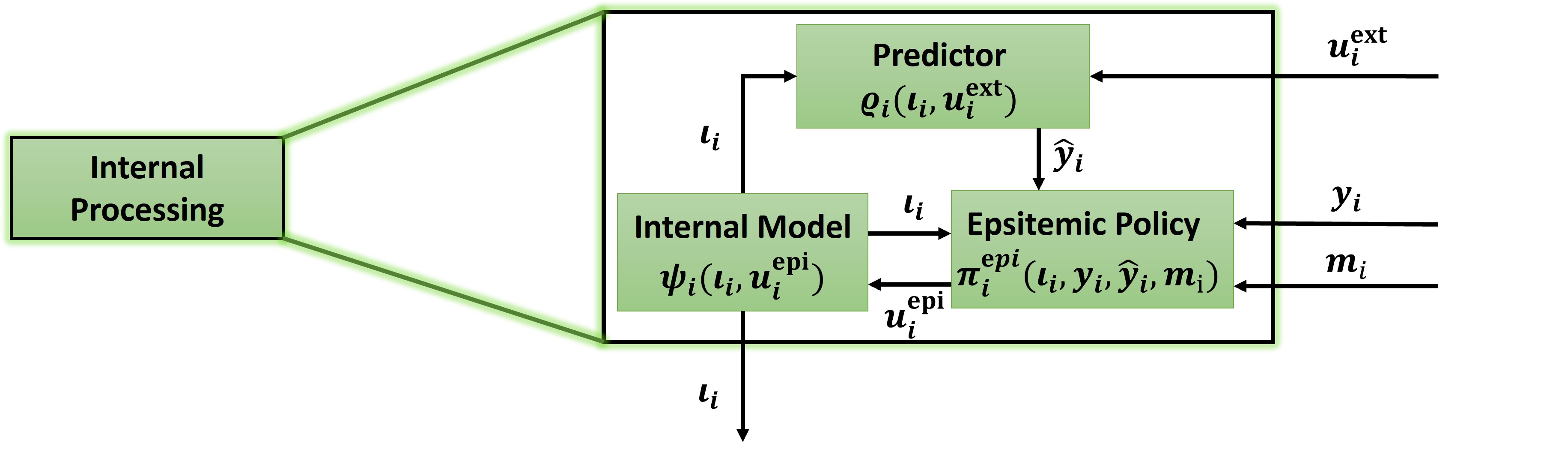}\label{fig:internal}}
  \hfill
    \caption{Proposed agent architecture illustrating two agents communicating with each other to refine their internal models. Dashed communicative arrows indicate that communication is not necessarily continuous and may occur as needed. Agent~$j$ has a similar structure for its `Ìnternal Processing' unit as the one depicted in (b).}
    \label{fig:system_model}
\end{figure*}

\textbf{Agent Architecture}. Building on this communication layer, we next specify the architecture of the agents. Inspired by the \emph{Information Transition System} (ITS) single-agent framework introduced in \cite{Basak, FRISTON1}, each agent is modeled as comprising two coupled subsystems: an \emph{external physical environment} and an \emph{internal information processing unit}, interacting via input-output channels. We extend the single-agent ITS framework in two key ways. First, the \emph{internal processing unit} now encompasses a structured internal model that captures the agent’s beliefs using epistemic logic and Kripke semantics, along with a predictor and epistemic policy to update these beliefs and maintain alignment with the environment. Second, we introduce a cooperative multi-agent ITS formulation where agents communicate to refine their internal models collectively. In contrast, the original single-agent ITS relied solely on a history of interaction trajectories to realize its internal model. To make these extensions concrete, we now turn to the details of the proposed architecture, which is depicted in Figure~\ref{fig:system_model}. The components are defined as follows:

\textbf{External Environment.} Let $\Xspace$ denote the external (physical) state space. Let $\UExti$ be its external action space. The joint action space is defined as $\UExt \;=\; \prod_{i=1}^{\agentsize}\UExti$,
so a joint external action $\uExt$ is an $\agentsize$-tuple $\uExt = (\uExto,\uExtt,\ldots,\uExtn)\in\UExt$. External dynamics are captured by
\[
\extmap:\Xspace\times \UExt\times \Dspace \to \Xspace,
\]
where $\Dspace$ models disturbances (stressors such as faults or attacks). Given the external state $\Xstate_t\in\Xspace$, a joint external action $\uExt_t\in\UExt$, and a disturbance $\dst_t\in\Dspace$, the next external state is $\Xstate_{t+1}=\extmap(\Xstate_t,\uExt_t,\dst_t)$. In the absence of disturbances, a default disturbance \(\dstd \in \Dspace\) is assumed so that the nominal dynamics are recovered; that is,  \(\extmap: \Xspace \times \UExt \rightarrow \Xspace\).
We note that the system is decentralized: at each time step $t$, agent $i$ selects its external action $\uExtit$ locally without centralized coordination. 


\medskip
\noindent\textbf{Perception and Communication.}
At each time step $t$, each agent $i$ gathers two types of inputs:
\begin{itemize}
  \item a physical observation $\Ystate_{i,t} = \senmap_i(\Xstate_t) \in \Yspace$ where $\Yspace$ denotes the space of observations, and $\senmap_i:\Xspace\!\to\!\Yspace$ is the sensor map of agent~$i$;
  \item a set of symbolic messages $(\msg_{j,t})_{j\in\neighbors_i}$ received from neighboring agents, conveying information about their current beliefs.
\end{itemize}
The precise logical form of these messages will be defined in Section~\ref{sec:kripke}, once the underlying epistemic language is introduced.

\medskip
\noindent\textbf{Internal Processing Unit.}
The internal processing unit encapsulates the agent’s belief representation, predictive capabilities, and epistemic adaptation. It comprises three key components: the internal model, a predictor, and an epistemic policy (see Figure~\ref{fig:internal}). We now describe each of these components in detail.

\medskip
\noindent\textit{Internal Model.}
Each agent maintains an internal representation of what it currently believes or considers possible about the environment. We denote agent~$i$’s internal state at time~$t$ by $\Istate_{i,t} \in \Ispace_i$, where $\Ispace_i$ is the set of admissible information (epistemic) states for that agent. For now, we view $\Istate_{i,t}$ abstractly as a structured summary of what the agent knows, does not know, and considers possible about the external environment. In Section~\ref{subsec:internal_kripke}, we make this notion concrete by instantiating $\Ispace_i$ with Kripke-style models. Under that interpretation, $\Istate_{i,t}$ corresponds to an accessibility structure that encodes which situations the agent regards as possible.


\medskip
\noindent\textit{Prediction, Epistemic Policy, and Internal Transition Map.}
To assess consistency between its internal model and the environment, agent~$i$ predicts the next observation using
\[
\varrho_i:\ \Ispace_i \times \UExti \to \Yspace,\qquad 
\Ystateh_{i,t+1} \;=\; \varrho_i\!\big(\Istate_{i,t},\,\uExt_{i,t+1}\big).
\]
A contradiction may arise between the predicted observation $\Ystateh_{i,t+1}$ and the actual observation $\Ystate_{i,t+1}$ if the prediction residual $\Upsilon\!\big(\,\Ystate_{i,t+1},\,\Ystateh_{i,t+1}\,\big)$ exceeds a pre-defined threshold $\varepsilon_1$ for at least $e_{\text{th}}$ time steps within the last $L$ steps, where $\Upsilon(\cdot,\cdot)$ is any distance measure on $\Yspace$. The agent updates its internal epistemic state $\Istate_{i,t}$ based on any raised contradiction, and any new information received via direct observation and/or messages from neighboring agents. These updates are modeled as epistemic actions
$u^{\mathrm{epi}}_{i, t+1} \in \UEps$, which operate on the epistemic state $\Istate_{i,t}$ to produce a revised state $\Istate_{i,t+1}$. Typical epistemic actions, collected in the set $\UEps$, include:
\begin{enumerate}
  \item \textsc{refine}: eliminate incompatible possibilities based on new information.
  \item \textsc{revise}: adjust the internal state to re-introduce new possibilities that accommodate changed conditions or evidence and resolve contradictions.
  \item \textsc{explore}: actively probe the environment to reduce uncertainty.
  \item \textsc{broadcast}: share selected information with neighbors.
  \item \textsc{hold}: leave the internal state unchanged when appropriate.
\end{enumerate}

\noindent These epistemic actions are selected by each agent $i$ according to an epistemic policy $\epPi$, which determines the appropriate epistemic action $u^{\mathrm{epi}}_{i, t+1}$ based on the agent’s current epistemic state $\Istate_{i,t}$, actual observations $\Ystate_{i,t+1}$, predicted observations $\Ystateh_{i,t+1}$, and received messages $(\msg_{j,t+1})_{j\in\neighbors_i}$. Let $\Mspace_i$ denote the set of (finite) message sequences agent~$i$ can receive (we defer their formal structure to Section~\ref{sec:kripke}). The \emph{epistemic policy} is a mapping
\[
\epPi:\ \Ispace_i \times \Yspace \times \Yspace \times \Mspace_i \to \UEps,
\]
that selects the action
\[
u^{\mathrm{epi}}_{i,t+1}
=\pi_i^{\mathrm{epi}}\!\big(\Istate_{i,t},\,\Ystate_{i,t+1},\,\Ystateh_{i,t+1},\,\Mspace_{i,t+1}\big).
\]
The internal transition map
\[
\psi_i:\ \Ispace_i \times \UEps \to \Ispace_i
\]
then produces the next internal state, $\Istate_{i,t+1}=\psi_i(\Istate_{i,t},\,u^{\mathrm{epi}}_{i, t+1})$.

Finally, to close the loop between beliefs and behavior, we define the external policy as follows.

\noindent\textbf{External Policy.}
Finally, agent $i$ converts its current epistemic state into a physical action via an an external policy
\[
\pi_i^{\mathrm{ext}}:\Ispace_i \to \UExti, 
\qquad \uExt_{i,t+1}=\pi_i^{\mathrm{ext}}(\Istate_{i,t}).
\]
Having established the agent architecture and its components, we now introduce the formal machinery we use to represent each agent’s internal model. Specifically, we present the \emph{temporal epistemic logic} framework and \emph{Kripke semantics} that underpin the structural foundation for agents' internal epistemic states and belief updates.

\section{Epistemic Logic and Semantic Communication}
\label{sec:kripke}

In the previous section, we introduced the agents’ internal models and described how their epistemic states evolve through prediction, observation, and communication. We now formalize these internal representations and communication mechanisms using tools from \emph{temporal epistemic logic} and \emph{Kripke semantics}. Temporal epistemic logic provides a symbolic framework for reasoning about what agents know, do not know, or consider possible \emph{over time}, while Kripke structures supply the semantic foundation that gives these notions a concrete interpretation. Within this framework, each agent~$i$’s internal model corresponds to a \emph{temporal-agent slice} of a Kripke structure that encodes its informational perspective on the world at a given moment. Communication among agents is then modeled \textbf{\textit{not merely as raw data exchange}}, but as the \textbf{\textit{transmission of logical statements}} whose truth values are evaluated with respect to these evolving Kripke structures; thus, capturing the \emph{semantic content} of messages rather than their syntactic form. 

This section first reviews the necessary background on temporal epistemic logic, introducing the logical language used to express formulas about knowledge and belief as they evolve over time. It then presents Kripke structures as formal models of knowledge and Kripke semantics as the rules that assign meaning to these logical formulas within such structures. Next, we define the agents’ internal epistemic spaces with respect to a Kripke structure, and finally describe how agent communication can be viewed as logic-based \emph{semantic} interaction.

\begin{figure}[t]
  \centering
  \subfloat[Environment setup.]{%
    \includegraphics[width=.48\linewidth]{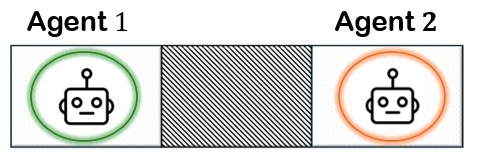}\label{fig:ex1env}}
  \hfill
  \subfloat[Kripke structure.]{%
    \includegraphics[width=.48\linewidth]{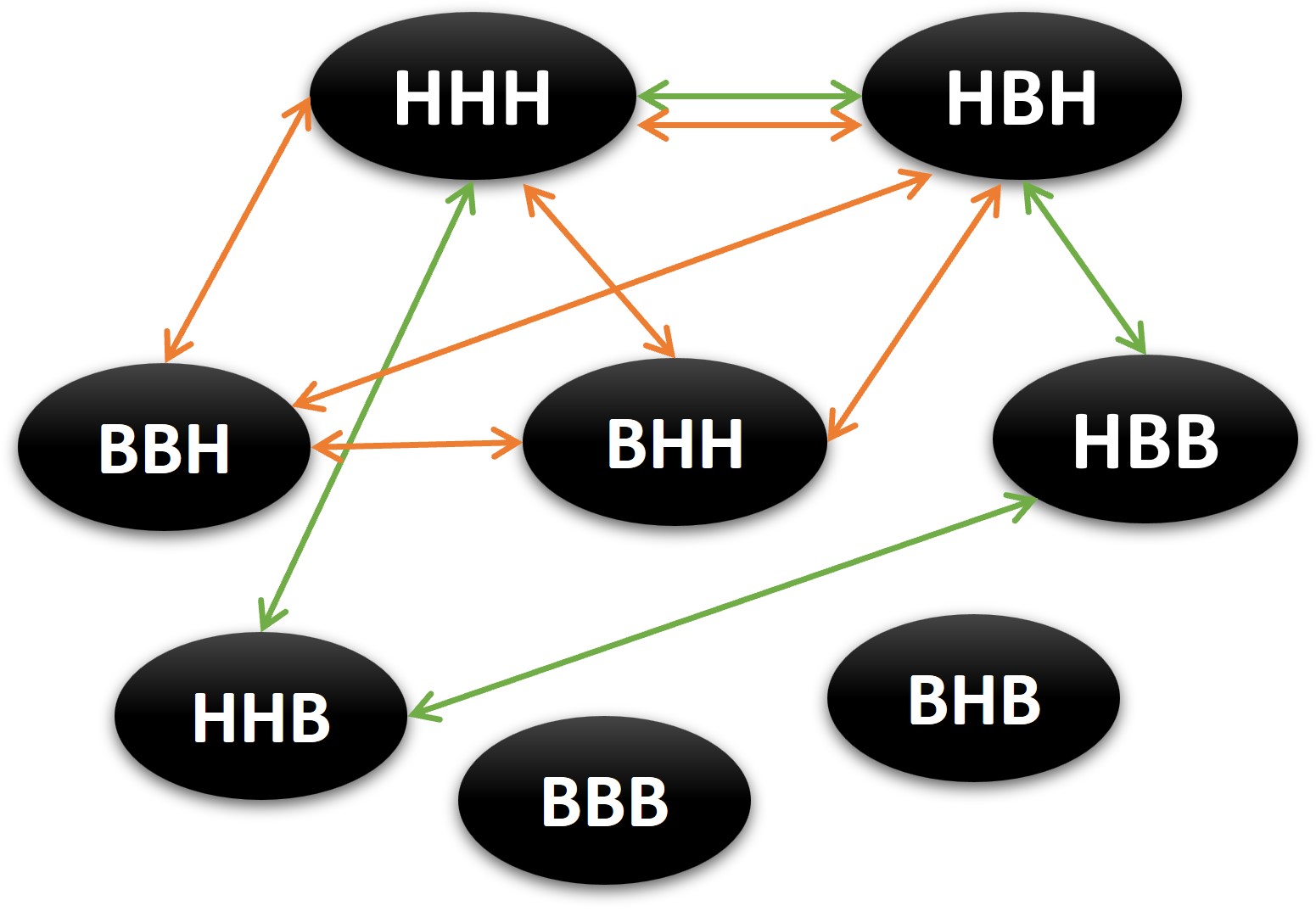}\label{fig:ex1kripke}}
    \caption{Example environment and its corresponding Kripke structure. In (a), agent~$1$ (green) is placed on cell~$1$, and agent~$2$ (orange) on cell~$2$. In (b), nodes represent possible worlds, with green and orange edges denoting the accessibility relations of agents~$1$ and~$2$, respectively. Reflexive arrows are omitted for clarity.}
    \label{fig:example1}
\end{figure}

\subsection{Background}
\label{subsec:background}
In multi-agent systems, the dynamic evolution of each agent’s beliefs about the environment is as important as the objective truths of the environment. Classical propositional logic evaluates statements as true or false but does not account for (i) different informational perspectives across agents (the epistemic aspect), or (ii) how those perspectives change over time (the temporal aspect). To address these limitations, we adopt a \emph{temporal epistemic logic} that combines multi-agent epistemic operators with temporal operators. This logic is built on a formal language for expressing formulas about knowledge and belief, and interpreted using Kripke structures and Kripke semantics.

\noindent\textbf{Atoms and Language.}
Let $\pSet$ be a nonempty set of atomic propositions, and let the discrete time domain be $\timeD=\{0,1,2,\ldots\}$. We also adopt the convention $\mathbb{N}=\{0,1,2,\ldots\}$. Although $\timeD$ and $\mathbb{N}$ denote the same underlying set, we use $\timeD$ for absolute time indices ($t \in \timeD$) and $\mathbb{N}$ (resp.\ $\mathbb{N}_{>0}$) for durations/step horizons; to keep timestamps distinct from lengths. With these conventions in place, the temporal–epistemic language $\lang_{\agentsize}^{\text{time}}(\pSet)$ is generated by the following grammar, where each expression $\formula$ denotes a well-formed formula of the language~\cite{TempoEpsLogic1}\cite{Epistemic1}:
\[
\formula \;::=\; \top \;\big|\; \p \;\big|\; \lnot \formula \;\big|\; \formula_1 \lor \formula_2 
\;\big|\; \glob_{[0,\beta)}\,\formula \;\big|\; \formula_1 \until_{[0,\alpha]} \formula_2 \;\big|\; \knows_i \formula,
\]
where the symbol $::=$ denotes a \emph{definition by syntactic formation},  commonly used in formal grammars to specify how well-formed formulas are constructed. It reads as ``is defined as'' or ``can be formed as,'' and the vertical bars $|$ separate alternative constructions. Moreover, the symbol $\top$ denotes a formula that is always true, $\p\in\pSet$, $\lnot$ denotes negation, $\lor$ denotes disjunction, $i\in\agentSet$, and $\alpha,\beta\in\mathbb{N}_{>0}$ denote step horizons. Standard connectives such as conjunction $\formula_1\land\formula_2$ and implication $\formula_1\!\rightarrow\!\formula_2$ are not primitive in the grammar but are defined using negation $\lnot$ and disjunction $\lor$ as follows:
\(\formula_1 \land \formula_2 \equiv \lnot(\lnot \formula_1 \lor \lnot \formula_2)\), and \(\formula_1 \rightarrow \formula_2 \equiv \lnot \formula_1 \lor \formula_2\). The above temporal and epistemic operators are interpreted as follows at time $t\in\timeD$:

\begin{itemize}
\item \textbf{Temporal \emph{Global} operator} $\glob_{[0,\beta)} \formula$: expresses that $\formula$ holds \emph{globally} over the bounded interval $[t,\,t+\beta)$.
  \item \textbf{Temporal \emph{Until} operator} $\formula_1 \until_{[0,\alpha]} \formula_2$: expresses that $\formula_1$ holds continuously \emph{until} $\formula_2$ becomes true, where $\formula_2$ must occur at some time $t'' \in [t,\,t+\alpha]$ and $\formula_1$ holds at all $t' \in [t,\,t'')$.
    \item \textbf{Epistemic \emph{Knowledge} operator}\footnote{While the syntax $\knows_i \formula$ contains no explicit time index, a formula $\formula$ is evaluated at the current time $t$ in our temporal setting. Intuitively, “agent $i$ knows $\formula$ at time $t$” means that, given the information available to $i$ at $t$, all situations $i$ still regards as possible make $\formula$ true. The formal satisfaction relation that makes this precise (including its time dependence) will be introduced when we define the Kripke structures and Kripke semantics.} $\knows_i \formula$: expresses that agent $i$ \emph{knows} $\formula$ at time $t$; that is, among \textbf{all} the possible scenarios consistent with agent $i$'s beliefs at time $t$, the formula $\formula$ is true in each one.
\end{itemize}

\noindent From the \textit{knowledge} operator $\knows_i \formula$, we can derive the dual operator which is the \emph{possibility} operator $\possible_i \formula \equiv \lnot \knows_i \lnot \formula$. This operator expresses that agent $i$ considers $\formula$ \emph{possible} at time $t$; that is,
there exist \textbf{some} scenario(s) that agent $i$ considers possible at time $t$ in which $\formula$ is true. Furthermore, let $\agentSet_{\text{0}} \subseteq \agentSet$ be a subset of agents. We can derive the \emph{ mutual knowledge operator} $\everyK_{\agentSet_{\text{0}}}  \formula$, which expresses the notion that ``everyone in group \(\agentSet_{\text{0}}\) knows \( \formula\)'' as follows:
\[
\everyK_{\agentSet_{\text{0}}}  \formula = \bigwedge_{i \in \agentSet_{\text{0}}} \knows_i  \formula.
\] 

\noindent\textbf{Kripke Structure.} To assign meaning to the formulas of the language $\lang_{\agentsize}^{\mathrm{time}}(\pSet)$ and to capture the evolution of agents’ beliefs over time, we employ a \emph{time-dependent Kripke structure} defined as~\cite{FHMV1995,Kripke1963SC}
\[
\kripke = \big(\worlds, \{\rel_{i,t}\}_{i \in\agentSet}, \val\big),
\]
where:
\begin{itemize}
  \item $\worlds$ is a nonempty set of \emph{possible worlds}, each representing a distinct situation (configuration) of the environment.
  \item For each agent $i \in \agentSet$ and time $t\in\timeD$, $\rel_{i,t} \subseteq \worlds \times \worlds$ is the \emph{epistemic accessibility relation} for agent~$i$ at time~$t$. Intuitively, if $(\worldw,\worldw')\in \rel_{i,t}$, meaning that $\worldw'$ is an \emph{accessible} world from $\worldw$, then at time $t$ agent $i$ cannot distinguish $\worldw'$ from $\worldw$; both are compatible with $i$’s information. The set of worlds accessible to an agent thus encodes the uncertainty in its knowledge at that time.
  \item $\val:\worlds \rightarrow 2^{\pSet}$ is the \emph{valuation function}, which specifies for each world $\worldw \in \worlds$ the set of atomic propositions that are true in $\worldw$.
\end{itemize}

\noindent\textbf{Kripke Semantics.} Given the Kripke structure $\kripke=(\worlds,\{\rel_{i,t}\}_{i\in\agentSet},v)$, the truth of a formula $\formula$ in the temporal–epistemic language $\lang_{\agentsize}^{\text{time}}(\pSet)$ is defined via a \emph{satisfaction relation} at evaluation points. An evaluation point is a triple $(\kripke,\run,t)$, where $\run:\timeD\to\worlds$ is a run; that is, a time-indexed trajectory
$\langle \run(0),\run(1),\ldots\rangle$ and $\worldw_t=\run(t)$ is the actual world at time $t$~\cite{temp_epis_run}. Intuitively, among all admissible trajectories, $\run$ is the actual one the system follows (e.g., the trace induced by the external dynamics and disturbances); hence, at each time $t$, it determines the current world $\worldw_t$. We write $(\kripke,\run,t)\models \formula$ to mean that \emph{$\formula$ is true at time $t$ along $\run$ in $\kripke$}. The satisfaction relation is then given inductively by the clauses below; taken together, these clauses constitute the \emph {Kripke semantics} for our temporal–epistemic language.


\begin{itemize}
  \item $(\kripke,\run,t)\models\top$ always holds.
  \item $(\kripke,\run,t)\models p$ iff $p\in v(\worldw_t)$.
  \item $(\kripke,\run,t)\models\lnot\formula$ iff $(\kripke,\run,t)\not\models\formula$.
  \item $(\kripke,\run,t)\models\formula_1\lor\formula_2$ iff $(\kripke,\run,t)\models\formula_1$ or $(\kripke,\run,t)\models\formula_2$.
  \item $(\kripke,\run,t)\models\glob_{[0,\beta)}\formula$ iff for all $t'$ with $t\le t'<t+\beta$ we have $(\mathcal{M}_{t'},\run,t')\models\formula$.
  \item $(\kripke,\run,t)\models\formula_1\until_{[0,\alpha]}\formula_2$ iff there exists $t''\in[t,t+\alpha]$ such that
        $(\mathcal{M}_{t''},\run,t'')\models\formula_2$ and for all $t'\in[t,t'')$ we have $(\mathcal{M}_{t'},\run,t')\models\formula_1$.
  \item $(\kripke,\run,t)\models\knows_i\formula$ iff for all $\worldw'\in\worlds$ with $(\worldw_t,\worldw')\in \rel_{i,t}$,
        \[
        (\kripke,\run[t\mapsto \worldw'],t)\models\formula,
        \]
        where $\run[t\mapsto \worldw']$ is the run that agrees with $\run$ at all times except that it maps $t$ to $\worldw'$.
\end{itemize}

\addb{In this paper, we will utilize a continuous example to elucidate the concepts under discussion.
\begin{example}[Kripke structure for a two-agent grid world]
\label{ex:example1}
Consider two agents placed in a $3\times 1$ grid with cells $c\in\{1,2,3\}$ (see Fig.~\ref{fig:ex1env}). Agent~$1$ (green) occupies cell~$1$, and Agent~$2$ (orange) occupies cell~$3$. Each agent $i$ observes only the color of its current cell $c$. Each cell is labeled by one of two mutually exclusive atomic propositions:
\[
H_c\ \text{ (“cell $c$ is white”)},\qquad B_c\ \text{ (“cell $c$ is black”)}.
\]
The set of all atomic propositions is $\pSet$. The environment can be in one of four configurations (i.e., possible worlds):
\begin{align}
\worlds=\big\{\,&\worldw^{(HHH)},\,\worldw^{(HBH)},\,\worldw^{(BHH)},\,\worldw^{(BBH)},\,\nonumber\\
&\worldw^{(HHB)},\,\worldw^{(HBB)},\,\worldw^{(BHB)},\,\worldw^{(BBB)}\big\} \nonumber
\end{align}
\newline
\noindent where the superscripts indicate the colors of cells $(1,2,3)$, e.g., $\worldw^{(HHB)}$ means cells~$1$ and $2$ are white and cell~$3$ is black.
The valuation function $\val$ lists which atoms are true in each world:
\[
\begin{aligned}
&v(w^{(HHH)})=\{H_1,H_2,H_3\},\, v(w^{(HBH)})=\{H_1,B_2,H_3\},\\
&v(w^{(BHH)})=\{B_1,H_2,H_3\},\, v(w^{(BBH)})=\{B_1,B_2,H_3\},\\
&v(w^{(HHB)})=\{H_1,H_2,B_3\},\, v(w^{(HBB)})=\{H_1,B_2,B_3\},\\
&v(w^{(BHB)})=\{B_1,H_2,B_3\},\, v(w^{(BBB)})=\{B_1,B_2,B_3\}.
\end{aligned}
\]
\noindent In this example, the actual world is assumed to be $w^{(HBH)}$. Figure~\ref{fig:ex1kripke} illustrates the corresponding Kripke structure: nodes represent possible worlds, and edges represent the accessibility relations $\{\rel_i\}_{i=1}^2$ that specify, for each agent $i$, its set of accessible worlds (the configurations that are consistent with agent $i$'s observations). Agent~$1$'s accessible worlds are $\{w^{(HHH)},w^{(HBH)},w^{(HBB)},w^{(HHB)}\}$, since these are the only worlds consistent with its observation. Because $H_1$ holds in \textbf{all} of Agent $1$'s accessible worlds, $\knows_1 H_1$ holds. Since Agent~$1$ does not know the color of cells $2$ and $3$, it considers the propositions \{$H_2$, $B_2$, $H_3$, $B_3$\} possible, as each holds in some of its accessible worlds. In other words, $\possible_1 B_2$, $\possible_1 H_2$, $\possible_1 H_3$, and $\possible_1 B_3$ hold. Likewise, since Agent~$2$ observes $c=3$, its accessible worlds are $\{w^{(HBH)},w^{(BBH)},w^{(HHH)},w^{(BHH)}\}$, and therefore $\knows_2 H_3$, $\possible_2 H_1$, $\possible_2 B_1$, $\possible_2 H_2$, and $\possible_2 B_2$ hold. 
\end{example}}


\subsection{Internal Epistemic Spaces}
\label{subsec:internal_kripke}
In Section~\ref{sec:system_model}, each agent’s internal state $\Istate_{i,t}$ was introduced as an abstract representation of what the agent knows or considers possible about the external environment. We now give this notion a formal grounding within the Kripke framework developed above. 

For each agent $i\in\agentSet$, the internal epistemic space is defined as
\[
\Ispace_i \;=\; \big\{\,\rel \subseteq \worlds \times \worlds \;\big|\; \rel \text{ satisfies chosen frame conditions}\,\big\},
\]
where $\rel_{i,t}\in\Ispace_i$ denotes the epistemic accessibility relation of agent~$i$ at time~$t$. Moreover, by \emph{frame conditions} we mean structural properties imposed on the accessibility relation to constrain epistemic behavior. Common examples include \emph{reflexivity} 
($\forall\,\worldw:\,(\worldw,\worldw)\in\rel_{i,t}$, meaning that the actual world is always possible) and \emph{transitivity} 
($( \worldw,\worldw')\in\rel_{i,t}$ and $(\worldw',\worldw'')\in\rel_{i,t}$ imply $(\worldw,\worldw'')\in\rel_{i,t}$, meaning indistinguishability composes). Intuitively, $(\worldw,\worldw')\in\rel_{i,t}$ means that, given its information at time~$t$, agent~$i$ considers world~$\worldw'$ possible when the actual world is~$\worldw$. Hence, $\rel_{i,t}$ captures the agent’s uncertainty or indistinguishability among worlds at that time. Thus, an agent’s internal state at time~$t$ is given by its accessibility relation,
\(
\Istate_{i,t} = \rel_{i,t} \in \Ispace_i,
\)
which encodes all worlds consistent with the agent’s current knowledge. Under this representation, $\Ispace_i$ provides the structural foundation for modeling how agents reason and update their knowledge over time.

\noindent
\textbf{Updates of the internal epistemic state.}
Recall from Section~\ref{sec:system_model} that agent $i$ updates its internal state in response to raised contradictions and new information received through observations and/or messages, via an epistemic action $u^{\mathrm{epi}}_{i,t+1}\in \UEps$. Under the Kripke representation, this becomes an update of the accessibility relation
\(
\Istate_{i,t}=\rel_{i,t}\ \in\ \Ispace_i
\) and is given by:
\[
\rel_{i,t+1}\;=\;\psi_i\!\big(\rel_{i,t},\,u^{\mathrm{epi}}_{i,t+1}\big),
\]
where $\psi_i:\Ispace_i\times \UEps\to \Ispace_i$ is the internal transition map. Typical actions instantiate $\psi_i$ as follows:

\begin{itemize}
  \item \textsc{Refine}: \emph{remove} accessibility relation to worlds inconsistent with new evidence, yielding a monotone decrease $\rel_{i,t+1}\subseteq \rel_{i,t}$.
  \item \textsc{Revise}:\ \emph{re-wire} accessibility relation among evidence-consistent worlds (may add new accessible worlds or remove existing ones) subject to minimal change.
  \item \textsc{Hold}/\textsc{Broadcast}/\textsc{Explore}:\ leave $\rel_{i,t}$ unchanged, inform neighbors, or trigger future evidence acquisition, respectively.
\end{itemize}
The update $\psi_i$ is required to (i) \emph{preserve frame conditions} (e.g., reflexivity, transitivity) so that $\rel_{i,t+1}\in\Ispace_i$; (ii) \emph{restore consistency} with the current evidence (observation at $t{+}1$ and received messages); and (iii) satisfy a \emph{minimal-change} principle (no unnecessary additions/removals of accessibility relations beyond what is needed for consistency). In this way, the abstract update policy of Section~\ref{sec:system_model} is realized concretely as an operation on Kripke accessibility relations.

\addb{\begin{remark}[Intuitive mapping between formal logic and operational algorithms]
\label{remark:connection}
To provide an intuitive link between the epistemic concepts introduced in Section~\ref{subsec:background} and the
operational procedures introduced here and used later in Section~\ref{sec:algorithms}, we briefly summarize the role of the Kripke
model in practical decision making. Each world represents a candidate
environmental condition, and the accessibility relation encodes which worlds
an agent still considers possible given its observations so far (referred to as accessible worlds). When new
information arrives that does not make any previously inaccessible world
possible, but instead rules out some of the currently accessible ones, the
agent removes exactly those worlds that are no longer compatible with the data
(\textsc{refine}). If, however, the agent's predictions under all currently
accessible worlds fail to explain the newly observed data, this produces a
\emph{prediction error} and prompts the agent to collect additional evidence
to discriminate between worlds (\textsc{explore}). As such evidence
accumulates and incompatible worlds are eliminated, the agent may reach a point
where its set of possible worlds must be reorganized; that is, when previously
inaccessible worlds become plausible again. In this case, the accessibility
relation is restructured (\textsc{revise}), and the agent communicates the
resulting belief change to its neighbors (\textsc{broadcast}). When neither
refinement nor revision is warranted, the agent simply keeps its current set
of possible worlds unchanged (\textsc{hold}). In this way, the abstract Kripke
semantics provides a conceptual basis for the detection, evidence-gathering,
revision, and communication behaviors implemented by the algorithms in later
sections.
\end{remark}}

\addb{\begin{remark}[Collisions as stressors]
Although collisions are not modeled explicitly in this paper, the framework can
naturally accommodate them. In this formalism, the behavior of other agents is
treated as part of the environment dynamics and is therefore encoded in each
agent’s epistemic state through the set of accessible worlds. If a collision
occurs (e.g., two agents select incompatible actions) the resulting observation
will contradict the predictions of all worlds in which such a collision should
not happen. This inconsistency is interpreted as a prediction error. The agent then updates its accessibility relation to retain only the worlds consistent with the observed
interaction, and its external policy realigns accordingly. In this way,
potential action collisions are handled as environmental stressors within the
same architecture developed in this section.
\end{remark}}

\subsection{Semantic Communication via Epistemic Logic and Kripke Structures}

Conventional communication systems primarily focus on the accurate transmission of symbols over a channel, optimizing metrics such as bit error rate and throughput. These systems treat information as a syntactic entity, largely independent of its meaning or relevance to the receiver. In contrast, semantic communication shifts the emphasis from symbol fidelity to meaning fidelity, prioritizing the successful conveyance of intended meaning and task-relevant information. Achieving semantic communication requires that (i) sender agents transmit compressed representations that preserve meaning, (ii) receiver agents interpret these messages within the appropriate context, and (iii) communication be driven by the task or objective at hand, rather than transmitting all available data indiscriminately. We realize semantic communication in a principled way using the temporal–epistemic language $\lang_{\agentsize}^{\mathrm{time}}(\pSet)$ and a shared Kripke structure with its semantics. Concretely:
(i) senders transmit messages that are logical formulas from $\lang_{\agentsize}^{\mathrm{time}}(\pSet)$, representing (possibly partial) snapshots of the sender’s epistemic state; that is, what the sender knows or considers possible.
(ii) receivers interpret the semantics (meaning) of these formulas by evaluating them in the shared Kripke structure, which provides the common context for correct interpretation, and
(iii) communication is task-driven, for example, when the task is \emph{achieving resilience}, agents communicate only to indicate that a change in the environment has occurred, allowing others to update their beliefs and policies accordingly. 
\begin{figure}[t]
  \centering
  \subfloat[Environment setup.]{%
    \includegraphics[width=0.6\linewidth]{images/Kripke_ex_env.jpg}\label{fig:a}}
  
  \par\medskip
  
  \subfloat[Agent 1's slice of $\mathcal{M}_0$ at $t=0$ (prior to communication).]{%
    \includegraphics[width=.48\linewidth]{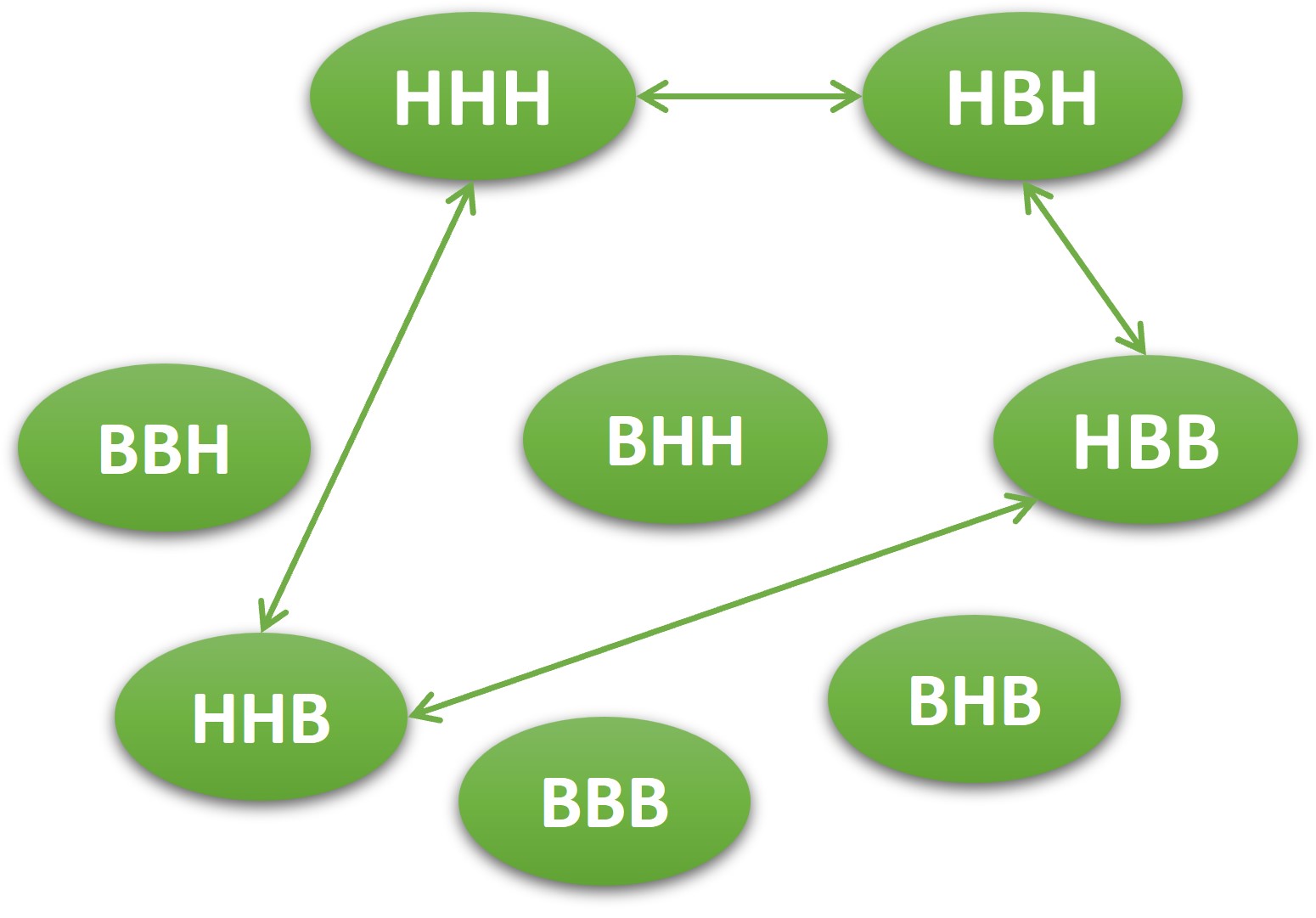}\label{fig:b}}
  \hfill
  \subfloat[Agent 2's slice of $\mathcal{M}_0$ at $t=0$ (prior to communication).]{%
    \includegraphics[width=.48\linewidth]{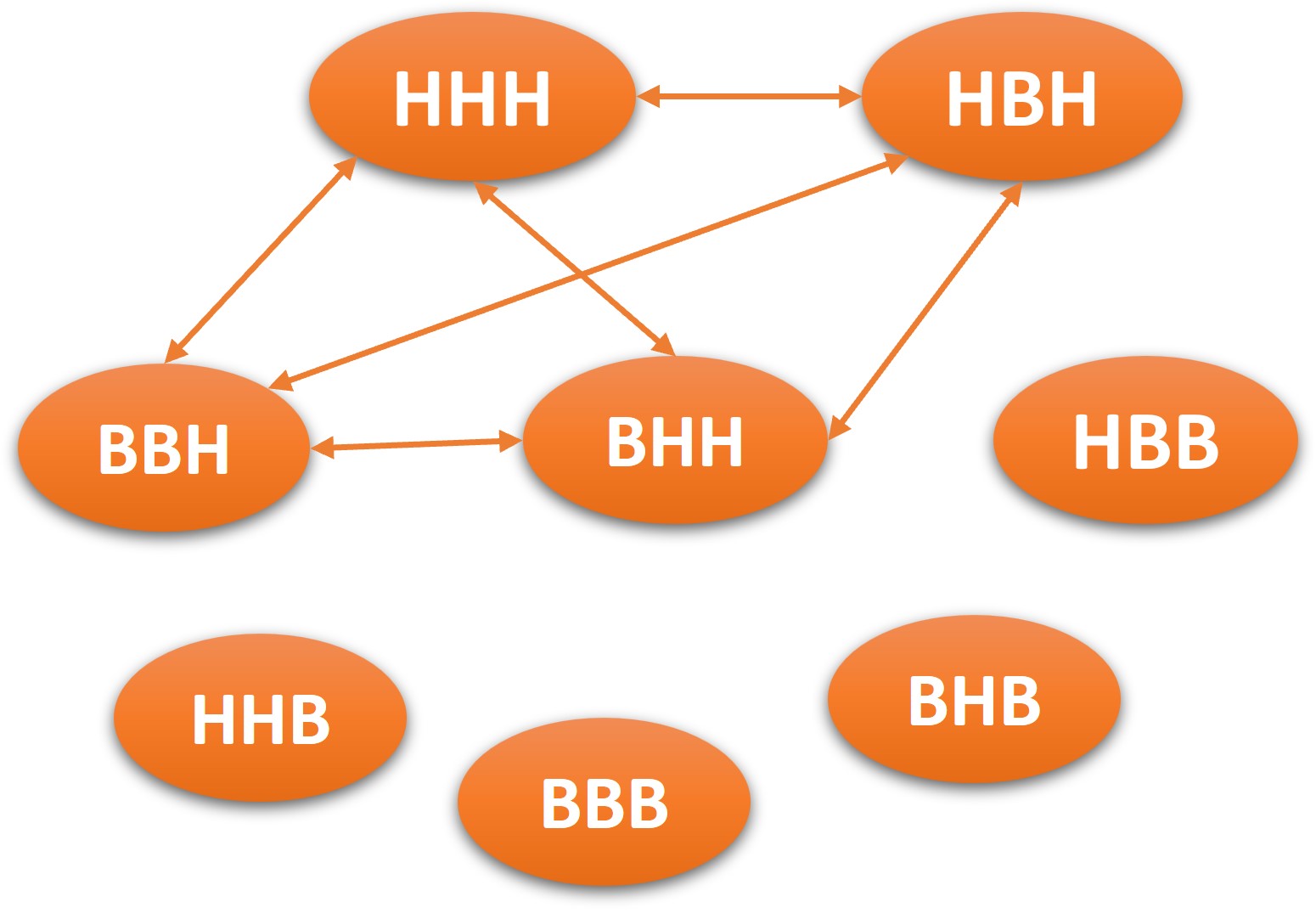}\label{fig:c}}

\par\medskip

  \subfloat[Agent 1's slice of $\mathcal{M}_1$ at $t=1$ (after communication).]{%
    \includegraphics[width=.48\linewidth]{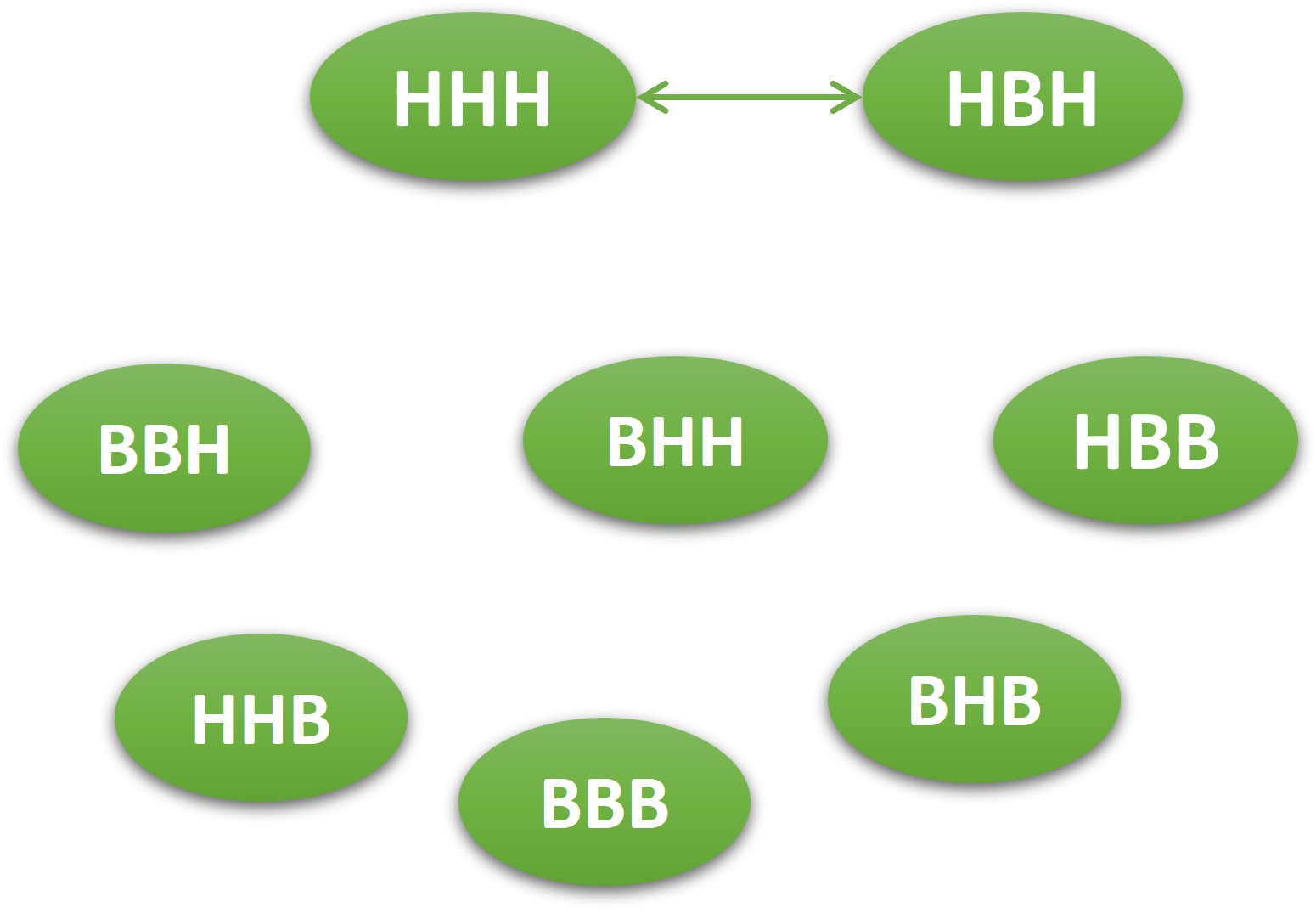}\label{fig:d}}
  \hfill
  \subfloat[Agent 2's slice of $\mathcal{M}_1$ at $t=1$ (after communication).]{%
    \includegraphics[width=.48\linewidth]{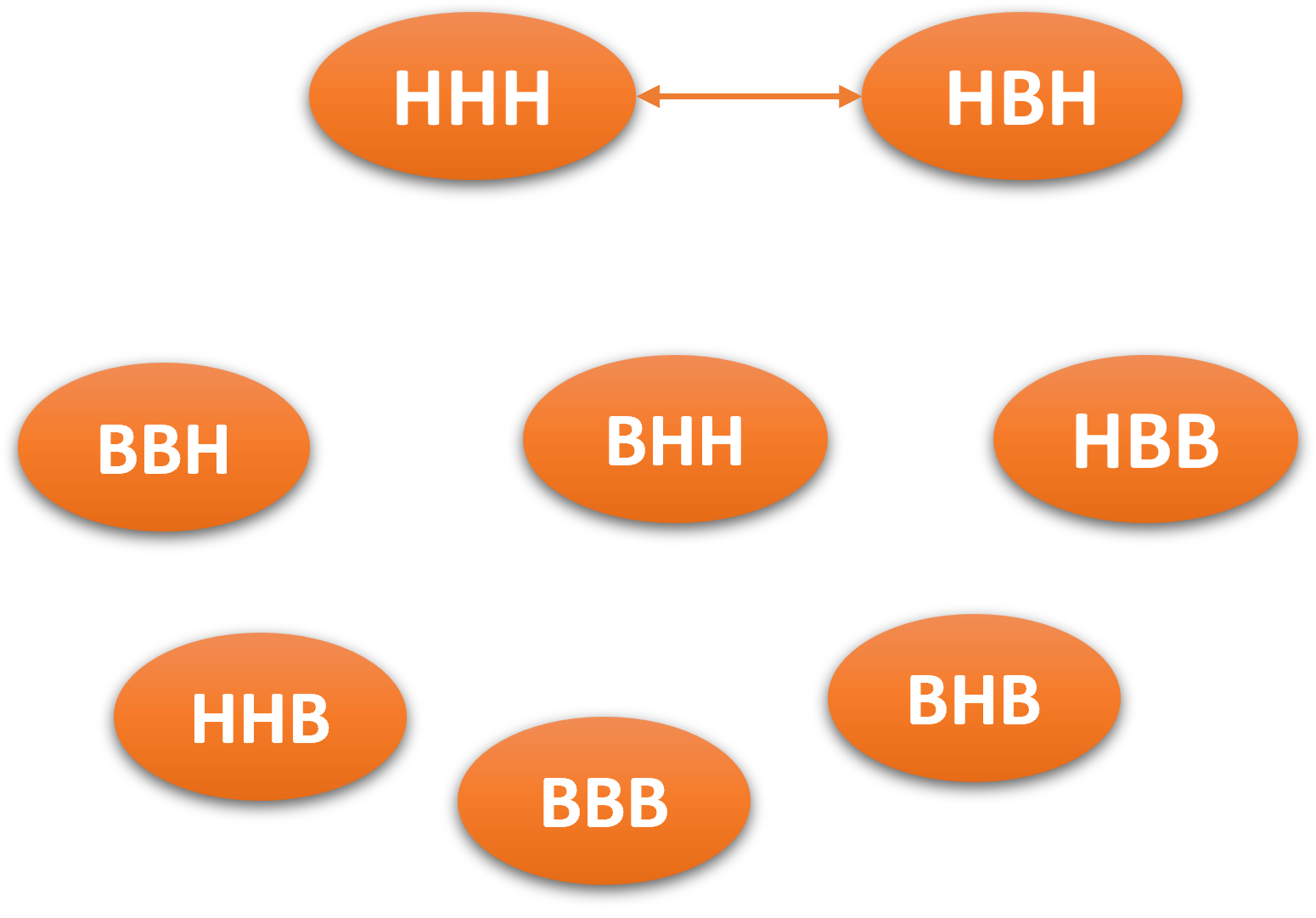}\label{fig:e}}
    
    \caption{Example environment and the corresponding per-agent slices of the shared Kripke structure, shown before and after communication. An agent’s slice is the triple $(\worlds,\rel_{i,t},\val)$ specific to agent $i$. Nodes represent possible worlds, and edges depict the agent-specific accessibility relations. Reflexive self-loops at worlds are omitted for clarity.}
    \label{fig:example2}
\end{figure}

\addb{
\begin{example}[Epistemic state refinement through communication]
\label{ex:example2}
We now extend Example~\ref{ex:example1} by introducing communication between agents. \textbf{Prior to communication}, each agent’s epistemic state is represented by its accessibility relation $\rel_{i,t}$ on the shared Kripke structure $\kripke$.
Figures~\ref{fig:b}--\ref{fig:c} illustrate the induced slices of the Kripke structure for each agent, showing only what the agent has locally (its own accessibility relation) rather than the entire set of relations in the shared structure. Within each slice, nodes correspond to possible worlds, and edges represent the agent’s accessibility relation, i.e., pairs of worlds that are indistinguishable to that agent based on its local observation. Worlds connected by these relations are called accessible worlds. Since each agent directly observes only one cell, many worlds remain accessible. Each slice is color-coded to match the corresponding agent. During communication, Agent~$1$ sends Agent~$2$ the formula $\knows_1 H_1$ and Agent~$2$ sends Agent~$1$ formula $\knows_2 H_3$. \textbf{After communication}, each agent refines its epistemic state by incorporating
the information received from the other agent. Formally, this corresponds to the epistemic action (\textsc{refine}) which  removes accessible worlds inconsistent with the shared message by deleting the corresponding accessibility relations. 
The resulting accessibility relations, shown in Figs.~\ref{fig:d}--\ref{fig:e},
are strictly smaller: each agent’s uncertainty is reduced, and both agents’ updated Kripke slices are consistent with the true world $\worldw^{HBH}$. Now all $\knows_1 H_1 \land \knows_1 H_3$, $\knows_2 H_1 \land \knows_2 H_3$, $\possible_1 H_2 \land \possible_1 B_2$, and $\possible_2 H_2 \land \possible_2 B_2$ hold.
\end{example}}

\begin{figure}[t]
  \centering
  \subfloat[\addb{Environment setup after an environmental change}.]{%
    \includegraphics[width=0.6\linewidth]{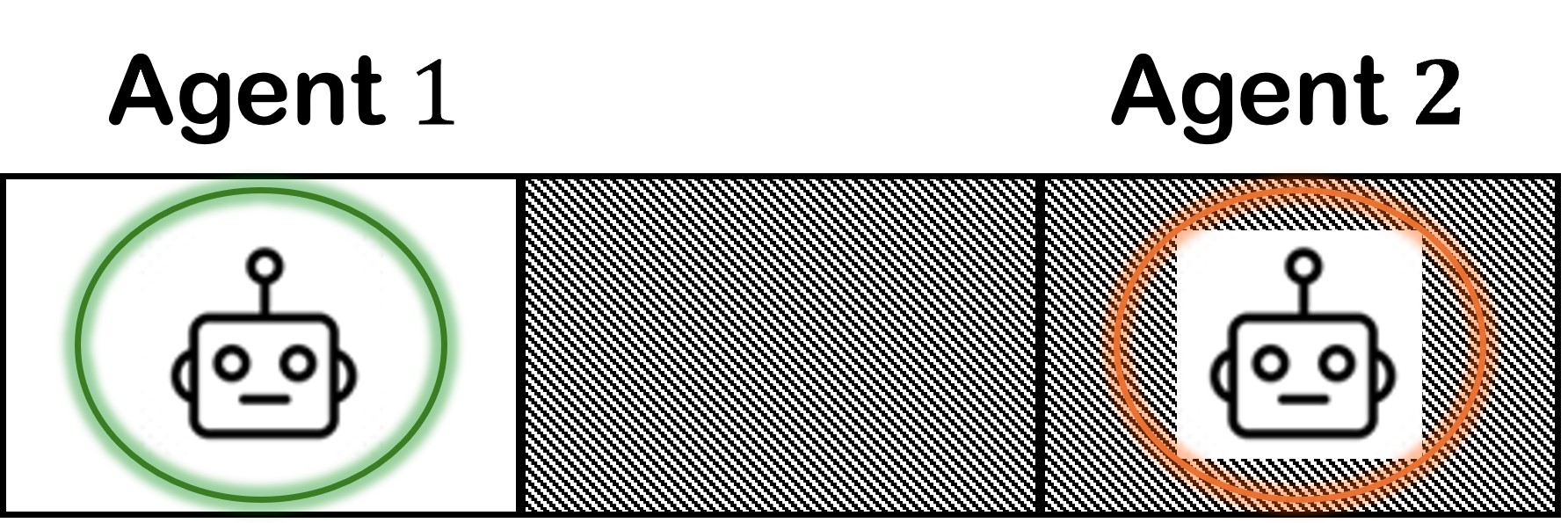}\label{fig:a_3}}
  
  \par\medskip
  
  \subfloat[\addb{Agent 1's slice of $\mathcal{M}_0$ after revision.}]{%
    \includegraphics[width=.48\linewidth]{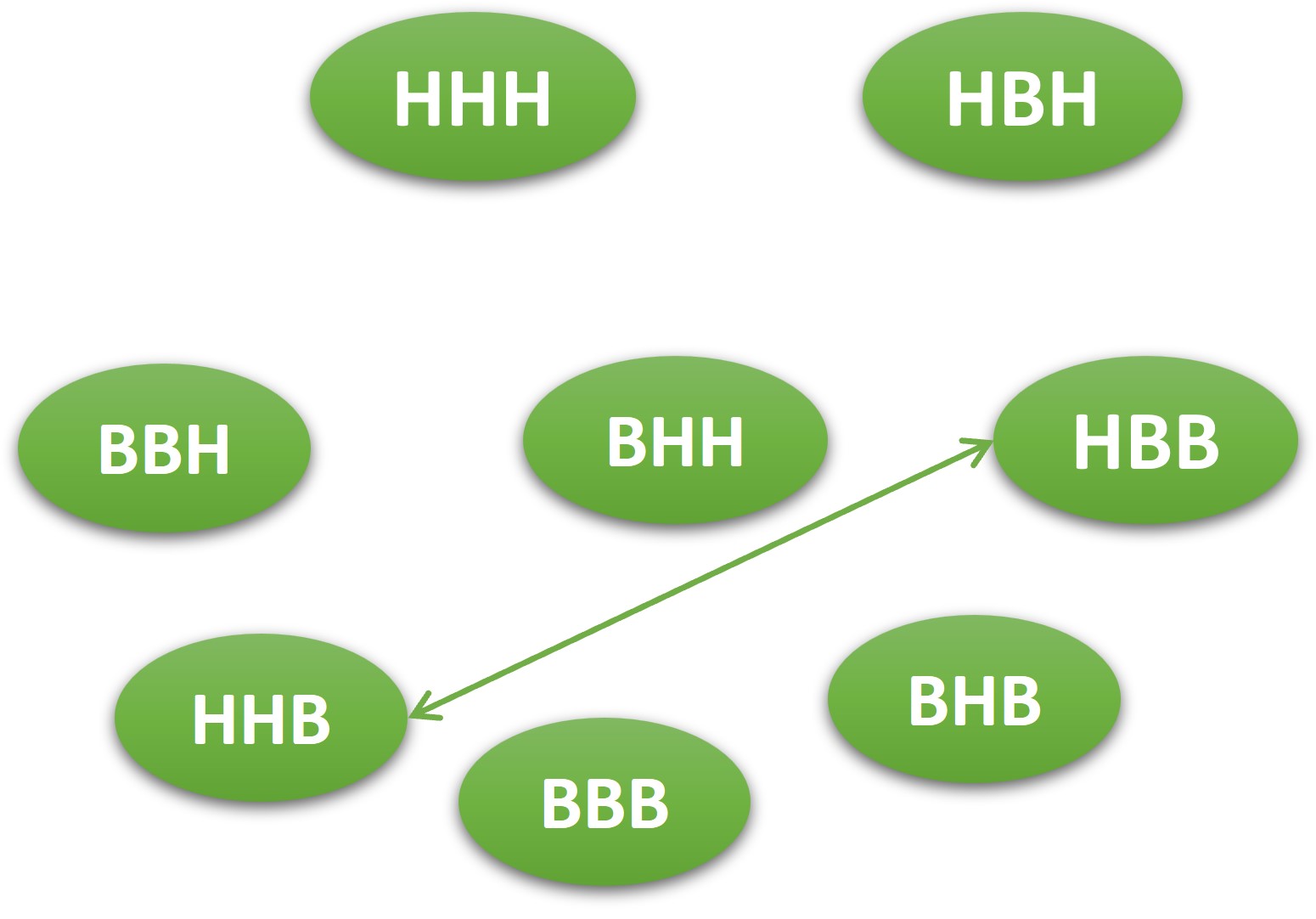}\label{fig:b_3}}
  \hfill
  \subfloat[\addb{Agent 2's slice of $\mathcal{M}_0$ after revision.}]{%
    \includegraphics[width=.48\linewidth]{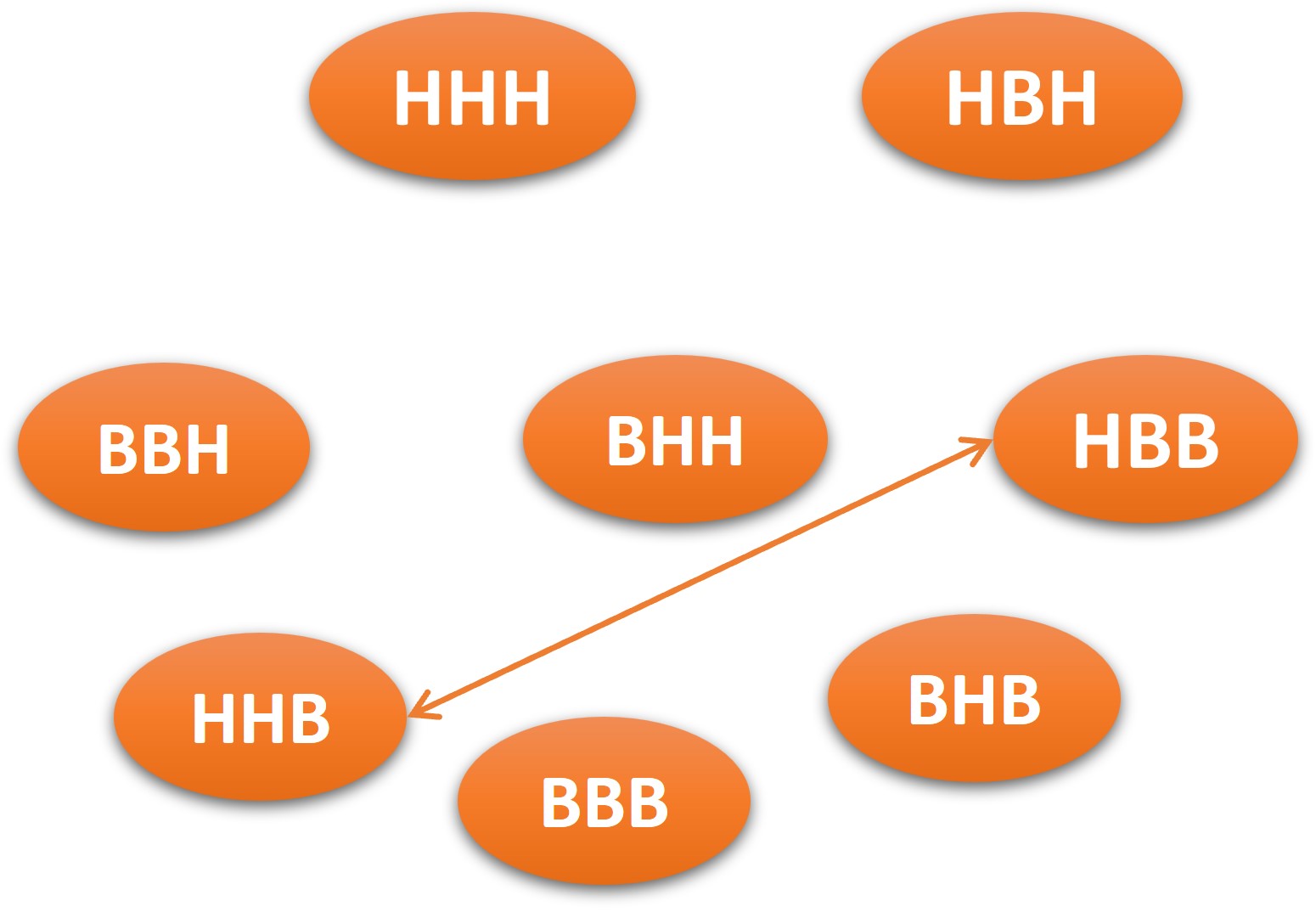}\label{fig:c_3}}

    \caption{\addb{Example environment after an environmental change (stressor) and the corresponding per-agent slices of the shared Kripke structure, shown after each agent executes the \textsc{revise} epistemic action. An agent’s slice is the triple $(\worlds,\rel_{i,t},\val)$ specific to agent $i$. Nodes represent possible worlds, and edges depict the agent-specific accessibility relations. Reflexive self-loops at worlds are omitted for clarity.}}
    \label{fig:example3}
\end{figure}

\addb{
\begin{example}[Epistemic state revision under a stressor]
\label{ex:example3}
We now consider an environmental change (stressor) affecting the setting in Examples~\ref{ex:example1} and~\ref{ex:example2}. Figure~\ref{fig:a_3} shows the environment after this change. Specifically, cell~$3$ transitions from white to black. Initially, Agent~$1$ receives no contradictory observations and therefore it executes the epistemic action \textsc{hold}. Meanwhile, Agent~$2$ locally detects a persistent mismatch between its predicted (via its knowledge) and observed color of cell~$3$. This triggers the epistemic action \textsc{explore}: Agent~$2$ collects discriminative observations to determine which environmental condition is now valid. Once sufficient evidence is collected, Agent~$2$ needs to update its epistemic state so that its knowledge is consistent with new evidence. This corresponds to executing the epistemic action \textsc{revise}, which updates Agent~$2$'s epistemic state so that $\knows_2 B_3$ holds and $\knows_2 H_3$ no longer holds. Concretely, Agent~$2$ removes accessibility links to worlds $HHH$ and $HBH$ (Figure~\ref{fig:e}), and add links to only those worlds that satisfy $B_3$, namely $HHB$ and $HBB$ (Figure~\ref{fig:c_3}). Agent~$2$ then communicates this update to Agent~$1$ via the formula $\knows_2 B_3$, which serves as an epistemic action \textsc{broadcast} that triggers a belief revision at Agent~$1$. In response, Agent~$1$ executes \textsc{revise} to remove previously accessible worlds violating $B_3$ (Figure~\ref{fig:d}) and re-align its epistemic state with the new publicly shared information (i.e., $\knows_1 B_3$ now holds and $\knows_1 H_3$ no longer holds). The resulting accessibility relations for Agent~$1$ is shown in Figure~\ref{fig:b_3}.
\end{example}}

\section{Resilience Metrics in Multi-Agent Systems}
\label{sec:metrics}
In critical situations, timely information and shared understanding are key to effective response. In our agent architecture, the success of agent $i$'s external policy $\pi_i^{\text{ext}}$ depends on the accuracy of its epistemic state $\Istate_{i,t}$, which the policy is conditioned on. Resilience thus requires agents to \textbf{detect} stressors, \textbf{recover} their epistemic states and policies after a stressor, and \textbf{maintain} accurate epistemic states and optimal policies post-recovery. To capture these aspects, we formalize resilience along two complementary dimensions:
\begin{itemize}
    \item \textbf{Epistemic resilience}: the ability to detect knowledge gaps, recover from discrepancies between the beliefs about the environment (epistemic state) and its actual state, and sustain accurate beliefs after recovery.
    \item \textbf{Action resilience}: the ability to re-align external actions with updated beliefs and maintain optimal external policies post-recovery.
\end{itemize}
These dimensions reflect not only individual capabilities but also the collective dynamics of multi-agent systems, where resilience is shaped by the formation of mutual knowledge; a shared, evolving understanding that enables coordinated adaptation and recovery. In light of this, we now formally introduce our resilience metric for multi-agent systems, followed by definitions of its two dimensions.

\begin{definition}[\textbf{System Resilience}]
\label{resS}
Let \(\resE, \resA \in \lang_{\agentsize}^{\text{time}}(\pSet)\) be temporal-epistemic formulas representing, respectively, epistemic resilience and action resilience specifications. Then, the system-wide resilience can be represented by the following formula: 
\begin{equation}
    \resS \equiv \resE \land \resA.
\end{equation}
This states that the overall system resilience is achieved if and only if both the epistemic and action resilience properties hold.
\end{definition}

\subsection{Epistemic Resilience}
Let \(\formula_1 \in \lang_\agentsize^{\text{time}}(\pSet)\) and \(\formula_2 \in \lang_\agentsize^{\text{time}}(\pSet)\) be mutually exclusive formulas (i.e., \(\formula_1 \rightarrow \lnot \formula_2\) and \(\formula_2 \rightarrow \lnot \formula_1\)) describing some aspect of the environment. Initially, \(\formula_1\) holds in the actual world, and all agents correctly believe that \(\formula_1\) is true. That is, at the initial time \(t_0\), for every agent \(i \in N\) and every \(\worldv \in W\) with \((\worldw,\worldv)\in \Istate_{i,t_0}\), we have \((\mathcal{M}_{t_0},\run,t_0)\models \formula_1\); equivalently,
\[
(\mathcal{M}_{t_0},\run,t_0)\models \bigwedge_{i\in N} \knows_i \formula_1 \quad\text{(i.e., }\everyK_\agentSet\,\formula_1\text{).}
\]

At some time $\tvio$, the environment changes (forming a stressor or a disturbance) such that \(\formula_1\) no longer holds and \(\formula_2\) becomes true. To achieve epistemic resilience, agents must: 
\begin{enumerate}[label=(\roman*)]
    \item Recover their epistemic states to the correct belief within a time bound \(\maxRE\), meaning that agents come to know that \(\formula_2\) now holds (\emph{recoverability}). That is, let recovery time be $\trecE$ where $\trecE - \tvio \le \maxRE$, for every agent \(i \in N\) and every \(\worldv \in W\) with \((\worldw,\worldv)\in \Istate_{i,\trecE}\), we have \((\mathcal M_{\trecE}, \run, \trecE)\models \formula_2\); equivalently,
\[
(\mathcal M_{\trecE},\run,\trecE)\models \bigwedge_{i\in N} \knows_i \formula_2 \quad\text{(i.e., }\everyK_\agentSet\,\formula_2\text{).}
\]
\item Maintain the recovered belief for at least a duration \(\minDE\) (\emph{durability}).
\end{enumerate}

The above conditions can be expressed more compactly in our temporal--epistemic language. 
In particular, epistemic resilience requires that all agents know the old fact $\formula_1$ until, within $\maxRE$ time units, they come to know the new fact $\formula_2$, 
and that this knowledge then persists for at least $\minDE$ time units. 
We formalize this as follows:

\begin{definition}[\textbf{Epistemic Resilience}]
Let \(\maxRE\) denote the maximum allowable recovery time and \(\minDE\) the minimum allowable durability duration. Let \(\resE\) denote a formula representing the epistemic resilience of the system. Then, we define epistemic resilience as
\begin{equation}
\label{resE}
    \resE \equiv \Bigl( \lnot \everyK_N \formula_2 \Bigr) \,\until_{[0,\maxRE]} \,\Bigl( \glob_{[0,\minDE)} \, \bigl( \everyK_N \formula_2 \bigr) \Bigr).
\end{equation}

where:
\begin{itemize}
  \item \(\until_{[0,\maxRE]}\) is the bounded until operator, ensuring that if the
agents are initially unaware of stressor, they will recover their epistemic states within $\maxRE$ time units;
  \item \(\glob_{[0,\minDE)}\) is the bounded globally operator, ensuring that
once the agents recover their epistemic states, these epistemic states are maintained for at least \(\minDE\) time units.
\end{itemize}
\end{definition} 
Durability means that the recovered epistemic states remain accurate enough to persist for at least \(\minDE\) time steps without having agents detecting any contradictions; assuming the environment does not change during this period. This highlights a fundamental \emph{recoverability–durability} trade-off: agents that revise their beliefs quickly may do so based on limited evidence, increasing the risk of incorrect updates and frequent revisions (i.e., lower durability). In contrast, agents that wait for more conclusive evidence may recover more slowly but are more likely to adopt a correct belief that remains stable (i.e., higher durability). To quantify these two factors, and make \eqref{resE} operational, we now extract two scalar quantities that
instantiate its two facets; \emph{epistemic recoverability} and \emph{epistemic durability}.

\begin{definition}[\textbf{Epistemic Recoverability}]
Let \(\timeV\) be the violation time at which the environment change (stressor) occurred, and let $\formula_2$ be the new environmental condition. Then, the \emph{epistemic recoverability interval} \(\Delta \trecE\) is defined as
\begin{equation}
\label{trecE}
    \Delta \trecE = \inf\{\,t \mid t \in \timeD,\; t > \timeV \text{ and } (\kripke,\run,t) \models \everyK_\agentSet \varphi_2\} - \timeV,
\end{equation}
This expression defines the epistemic recoverability duration as the smallest interval following the violation time $\timeV$ after which the agents' epistemic states satisfy the property that every agent in $\agentSet$ knows $\formula_2$. We have the convention that $\Delta \trecE= +\infty$ if the set is empty (i.e., epistemic accuracy is never restored after $\tvio$). The time at which this recovered epistemic property is first established is then given by:
\(
\trecE = \timeV + \Delta \trecE.
\)
\end{definition}

\begin{definition}[\textbf{Epistemic Durability}]
Let \(\timeV\) be the violation time at which the environment change occurred. Then, the \emph{epistemic durability interval} \(\Delta \tdurE\) is defined as
\begin{equation}
\label{tdurE}
\Delta \tdurE = \inf\{ t \mid t \in \timeD, t > \trecE \text{ \& } (\kripke,\run,t) \not\models \everyK_\agentSet \formula_2 \} - \trecE,
\end{equation}
where \(\everyK_N \formula_2\) denotes that every agent in \(N\) believes \(\formula_2\), and \((\kripke, \run,t) \not\models \everyK_\agentSet \formula_2\) indicates that this property no longer holds at time \(t\) in world \(\worldw\) and that agents have changed their beliefs. In other words, \(\Delta \tdurE\) is the smallest time interval after \(\trecE\) during which the recovered belief persists before being violated. We have the convention that $\Delta \tdurE =+\infty$ if the set is empty (i.e., epistemic accuracy is never violated after $\trecE$).
\end{definition}

While epistemic resilience captures the agents’ ability to recover and maintain accurate beliefs about the environment, resilience does not end at the epistemic level. For agents operating in dynamic environments, knowledge must be translated into appropriate external actions. Even if agents correctly revise their beliefs, a system is not truly resilient unless these beliefs lead to timely and effective adjustments in behavior. We therefore turn to the second dimension: 
\emph{action resilience}.

\subsection{Action Resilience}
\addb{Let $\Pi_i^{\mathrm{ext}}$ denote the set of admissible external policies for
agent $i$ and let $\pi_{i,t}^{\mathrm{ext}} \in \Pi_i^{\mathrm{ext}}$
denote the external policy currently implemented by agent $i$ at time $t$. Each agent $i$ is associated with a specific task-related performance objective $J_i$ (e.g., expected cumulative reward or cost). Given its finite candidate set of external policies
$\Pi_i^{\mathrm{ext}}$, this objective induces a time-dependent set of \emph{individually
optimal external policies}
\[
\Pi_{i,t}^{\mathrm{opt}} := \arg\max_{\pi_i \in \Pi_i^{\mathrm{ext}}} J_i(\pi_i; t),
\]
which contains the policies that are optimal for agent $i$ at time $t$ under the current environment condition. At the logical level, we introduce for each agent $i$ an atomic proposition
$\pi_i^{\mathrm{opt}} \in \lang_\agentsize^{\text{time}}(\pSet)$ that is interpreted as ``agent $i$ is
currently acting optimally with respect to $J_i$’’. That is, the proposition $\pi_i^{\mathrm{opt}}$ holds true iff agent $i$'s
current external policy $\pi_{i,t}^{\mathrm{ext}}$ belongs to $\Pi_{i,t}^{\mathrm{opt}}$. The group-level predicate is the conjunction
\[
\pi^{\mathrm{opt}} \equiv \bigwedge_{i\in\agentSet} \pi_i^{\mathrm{opt}},
\]
which holds precisely when every agent executes one of its individually
optimal external policies. No joint optimality or equilibrium concept is
assumed; optimality is evaluated individually for each agent.}



Initially, before the stressor occurred at time step $\tvio$, all agents executed optimal actions with respect to the environment conditions. That is, at times \(t'< \tvio\), for every agent \(i \in N\) and every \(\worldv \in W\) with \((\worldw,\worldv)\in \Istate_{i,t'}\), we have \((\mathcal{M}_{t'},\run,t')\models \pi^{\text{opt}}\). Once the stressor occurs at time $\tvio$ and cause a change in agents' epistemic states at time $\trecE$, agents' external policies may no longer be optimal. That is, at $t'' \ge \tvio$, for every agent $i$ and every world $\worldw' \in \worlds$ with \((\worldw,\worldv)\in \Istate_{i,t''}\), we have \((\mathcal{M}_{t''},\run,t'')\not\models \pi^{\text{opt}}\). To achieve action resilience, agents must:
\begin{enumerate}[label=(\roman*)]
\item Recover their external policies to the optimal ones within a time bound \(\maxRA\), meaning that \(\pi^{\text{opt}}\) now holds (\emph{recoverability}). That is, let recovery time be $\trecA$ where $\trecA - \trecE \le \maxRA$, for every agent \(i \in \agentSet\) and every \(\worldv \in W\) with \((\worldw,\worldv)\in \Istate_{i,\trecA}\), we have \((\mathcal{M}_{\trecA}, \run,\trecA)\models \pi^{\text{opt}}\).
\item Maintain optimal policies for at least a duration \(\minDA\) (\emph{durability}).
\end{enumerate}

Similar to epistemic resilience, the above conditions can be expressed more compactly in our temporal--epistemic language. 
In particular, action resilience requires that all agents do not execute optimal policies with respect to current environment conditions ($\lnot \pi^{\text{opt}}$ holds)
until, within $\maxRE$ time units, they recover to execute optimal policies ($\pi^{\text{opt}}$ holds), and that they maintain executing optimal policies for at least $\minDA$ time units. We formalize this in the following definition.


\begin{definition}[\textbf{Action Resilience}]
Let \(\maxRA\) denote the maximum allowable recovery time for agents to adjust their external actions based on the recovered collective knowledge, let \(\minDA\) denote the minimum duration for which the optimal action must be maintained, and let $\pi^{\text{opt}}_i$ be a proposition indicating that agent $i \in \agentSet$ is currently acting optimally. Let \(\resA\) denote the action resilience of the system. Then, we define action resilience as
\begin{equation}
\label{resA}
    \resA \equiv \Bigl( \neg \pi^{\text{opt}} \Bigr) \,\until_{[0,\maxRA]} \,\Bigl( \glob_{[0,\minDA)} \, \bigl( \pi^{\text{opt}} \bigr) \Bigr),
\end{equation}

where:
\begin{itemize}
    \item \(\until_{[0,\maxRA]}\) is the bounded \textit{until} operator, ensuring that if the agents are initially not acting optimally, they will recover optimal actions within \(\maxRA\) time units;
    \item \(\glob_{[0,\minDA)}\) is the bounded \textit{globally} operator, ensuring that once the optimal actions are adopted, they are maintained for at least \(\minDA\) time units.
\end{itemize}
\end{definition}

This formulation reflects a \emph{recoverability–durability} trade-off: agents may choose rapid action adjustments (leading to frequent changes and lower durability) or explore enough before revising their policies (ensuring longer-term stability). As done previously, to quantify these two factors, and make \eqref{resA} operational, we now extract two scalar quantities that
instantiate its two facets; \emph{action recoverability} and \emph{action durability}.

\begin{definition}[\textbf{Action Recoverability}]
Let \(\timeV\) be the violation time at which the environment change occurred. Then, the \emph{action recoverability interval} \(\Delta \trecA\) is given by
\begin{equation}
\label{trecA}
    \Delta \trecA = \inf\{\, t  \mid t \in \timeD,\; t \ge \trecE \text{ and } (\kripke, \run,t) \models \pi^{\mathrm{opt}} \} - \trecE.
\end{equation}
This definition states that $\Delta \trecA$ is the smallest time interval after $\trecE$ until the system reaches the state where all agents in $\agentSet$ act optimally. We have the convention $\Delta \trecA=+\infty$ if the set is empty (i.e., optimality is never reached after $\trecE$). 
\end{definition}

\begin{definition}[\textbf{Action Durability}]
    Let \(\timeV\) be the violation time at which the environment change occurred. Define 
    \[
    \trecA = \trecE + \Delta \trecA,
    \]
    to be the time at which the optimal policy is reached. Then, the \emph{action durability interval} $\Delta \tdurA$ is given by
    \begin{equation}
    \label{tdurA}
        \Delta \tdurA = \inf\{ t  \mid t \in \timeD, \, t \geq \trecA \text{ and } (\kripke, \run,t) \not\models \pi^{\mathrm{opt}} \} - \trecA. 
    \end{equation}
    This definition states that $\Delta \tdurA$ is the minimal time interval after $\trecA$ until the system no longer satisfies the condition that every agent is acting optimally. In other words, it measures how long the recovered, optimal action is maintained before being violated. We have the convention $\Delta \tdurA =+\infty$ if the set is empty (i.e., optimality is never violated after $\trecA$). 
\end{definition}

Next, we present policies and algorithms that achieve epistemic and action resilience in multi-agent systems under the proposed metrics.

\section{Achieving Epistemic and Action Resilience: Policies and Algorithms}
\label{sec:algorithms}
Having defined the resilience metrics, we now specify mild system assumptions under which
information can flow and stressors are detectable, and then present policies and algorithms
that realize the metrics.

\begin{figure*}[ht]
    \centering
    \includegraphics[width=0.9\textwidth]{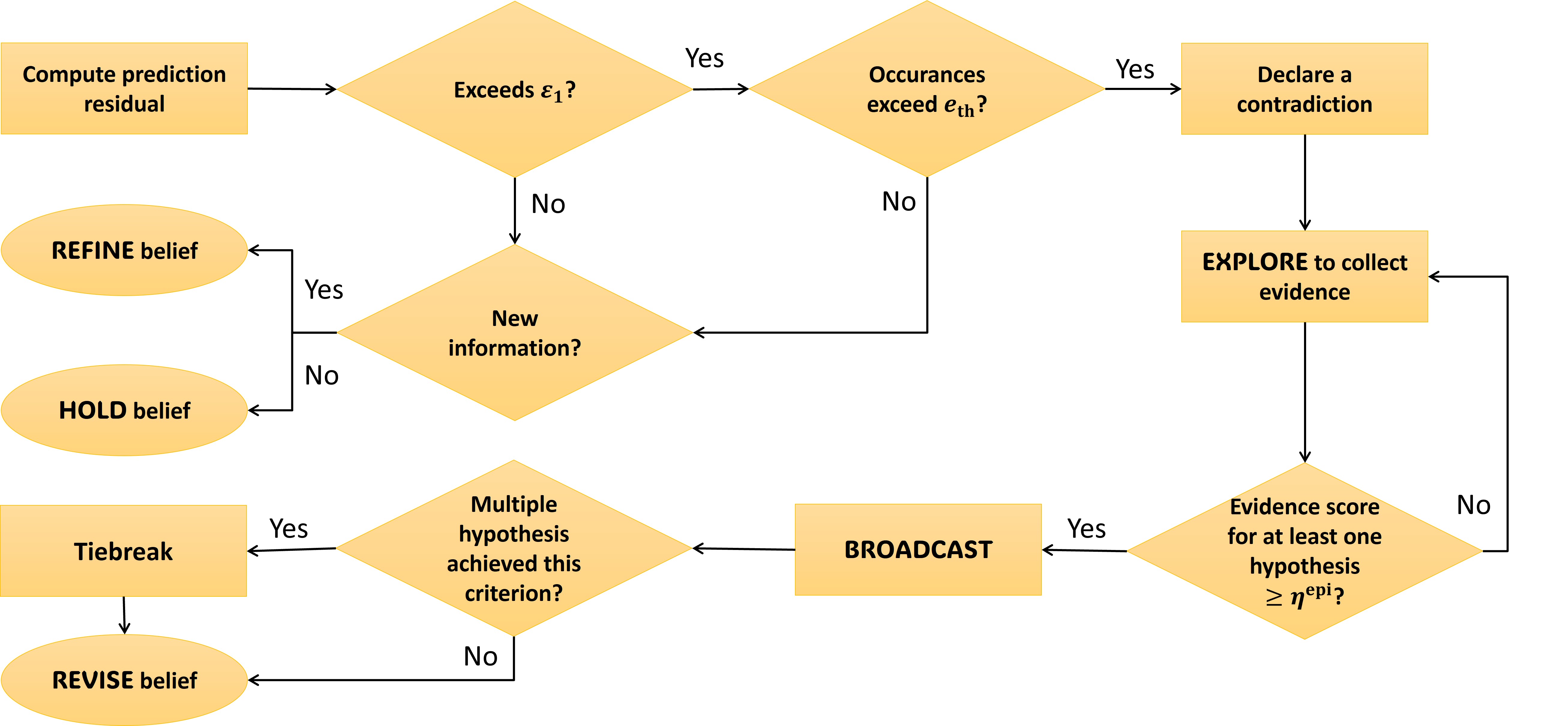} 
    \caption{Flow diagram of the epistemic policy executed by every agent to achieve epistemic resilience.}
    \label{fig:epi_policy}
\end{figure*}

\begin{assumption}[\textbf{Communication Network Connectivity and Reliability}]
\label{assump:communication}
The communication network is a directed graph \(\graph=(\agentSet,\edges)\) that is
\emph{strongly connected}: for any \(i,j\in\agentSet\) there exists a directed path
\(i\leadsto j\). Define the (hop) diameter
\[
\lcomm \;:=\; \max_{i,j\in\agentSet} \link(i,j),
\]
where \(\link(i,j)\) is the length of a shortest directed path from \(i\) to \(j\). \addb{We additionally assume that messages transmitted along edges in \(\edges\) are delivered reliably and within one time step.} Thus, \(\lcomm\) serves as an upper bound on the time required to disseminate a contradiction or belief update to the entire network, thereby enabling the formation of mutual knowledge \(\everyK_\agentSet(\cdot)\).
\end{assumption}

\begin{assumption}[\textbf{Timely Local Detection of Contradictions}]
\label{assump:contradiction_detection}
Let \(\formula_1\) denote a proposition that becomes false in the environment at some time \(\tvio\). We assume that there exists at least one agent \(i \in \agentSet\) that can detect this change locally within at most \(\ndct\) time steps. Formally, 
\[\exists i \in \agentSet, \,\exists \tdet \in [\tvio, \tvio + \ndct] \text{ s.t. } ({\mathcal M}_{\tdet}, \run,\tdet) \models \knows_i \lnot \formula_1.\]
The detection is assumed to be governed by an internal inconsistency condition, such as a threshold-based deviation between predicted and observed observations.
\end{assumption}

Assumption~\ref{assump:contradiction_detection} states that at least one agent can \emph{notice} a stressor from its own predicted observation and the actual observation mismatch. This is a reasonable assumption: if a stressor has no observable effect on any agent’s knowledge (no persistent residual beyond noise) or on task performance, then it is 
epistemically silent at that time and there is no ``meaningful recovery'' to perform.

\begin{assumption}[\textbf{Minimum Inter-Disturbance Interval}]
\label{assump:freq_dst}
Let \(\{\tvio^{(k)}\}_{k\ge1}\) be the sequence of disturbance (stressor) times.
There exists \(\drate>0\) such that for all \(k\),
\[
\tvio^{(k+1)} - \tvio^{(k)} \;\ge\; \drate.
\]
Equivalently, if a change occurs at time \(\tvio\), any subsequent change occurs at $\tvio'$ should satisfy
\(\tvio' \ge \tvio + \drate\). Thus, disturbances occur no more frequently than every $\drate$ time units.
\end{assumption}

Assumption~\ref{assump:freq_dst} prevents adversarially rapid, 
back-to-back stressors that would make recovery provably impossible. Feasibility, however, also requires that this inter-disturbance gap be large enough to
accommodate worst-case recovery and durability windows. The following bound captures this
necessary time budget.
\begin{remark}[Necessary time-budget for resilience]
\label{rem:time_budget}
Under Assumption~\ref{assump:freq_dst}, a necessary condition for the resilience
specification to be satisfiable is
\[
\drate \;>\; \max\big\{\ \maxRE+\minDE,\ \ \maxRE+\maxRA+\minDA\ \big\}.
\]
Indeed, within any inter-disturbance interval one must (i) complete epistemic recovery (at most $\maxRE$) and maintain it for at least $\minDE$, and (ii) after epistemic recovery (at most $\maxRE$), complete action recovery (at most $\maxRA$) and maintain it for at least $\minDA$; hence the interval must exceed both sums. This condition is \emph{necessary but not sufficient}: satisfying it does not by itself
guarantee resilience; appropriate epistemic/action policies and their success
properties are also required.
\end{remark}

For the reminder of this section, we work under Assumptions~\ref{assump:communication}–\ref{assump:freq_dst} and we assume that the time-budget bound in Remark~\ref{rem:time_budget} holds.

\subsection{Achieving Epistemic Resilience}
\label{subsec:achieve_epi_rec}
Let \(\formula_1 \in \lang_\agentsize^{\text{time}}(\pSet)\) be a formula representing an environmental condition that holds true prior to the violation time $\timeV$, but becomes false at $\timeV$. Let $\Phi= \{\varphi_2, \cdots, \varphi_\kpros\}$ be a set of $\kpros-1$ mutually exclusive hypotheses, each generated from the language $\lang_\agentsize^{\text{time}}(\pSet)$. Each hypothesis characterizes the same environmental aspect, with one hypothesis describing the new condition following a disturbance. Since these hypotheses are mutually exclusive, they cannot be true simultaneously.

To enable agents to detect such environmental changes (that is, to distinguish when the previously valid condition $\formula_1$ no longer holds and a new hypothesis in $\Phi$ becomes true) we quantify discrepancies between the agents’ internal predictions and the actual environment. As stated in Section~\ref{sec:system_model}, we introduce a measure of prediction inconsistency based on observation residuals. In particular, for any agent \(i \in \agentSet\), the prediction residual is $\Upsilon\!\big(\,\Ystate_{i,t},\,\Ystateh_{i,t}\,\big)$, with $\Ystate_{i,t}=\senmap_i(\Xstate_t)$, $\Ystateh_{i,t}=\varrho_i(\Istate_{i,t-1},\uExtit)$, and \(\Upsilon\) is any distance measure on \(\Yspace\) (e.g., the \(\ell_2\)-norm). Fix a threshold \(\varepsilon_1>0\), a window length \(L\in\mathbb{N}_{\ge 1}\), and an exceedance-count threshold \(e_{\text{th}}\in\{1,\ldots,L\}\). Define the count of threshold exceedances in the last $L$ steps as:
\begin{equation}
\label{eq:threshold}
e_{i,t} \;:=\; \sum_{t'=t-L+1}^{t} \mathbf{1}\{\, \Upsilon\!\big(\,\Ystate_{i,t'},\,\Ystateh_{i,t'}\,\big) > \varepsilon_1 \,\}.
\end{equation}
A contradiction is declared by agent $i$ at time $t$ precisely when
\(
e_{i,t}\ \ge\ e_{\text{th}}.
\)
\addb{This exceedance-count mechanism ensures that isolated noisy readings or transient prediction jitters do not trigger contradiction detection; only persistent, statistically meaningful deviations between predicted and sensed observations initiate evidence gathering and epistemic revision.} If no contradiction is detected (i.e., $e_{i,t} < e_{\text{th}}$), agent~$i$ selects either \textsc{refine} or \textsc{hold}, depending on whether new information has arrived via direct observations or messages (if any). Otherwise, agent \(i\) declares a contradiction at the first time $t$ with at least $e_{\text{th}}$ exceedances in the last $L$ steps ((i.e., $e_{i,t} \ge e_{\text{th}}$)), namely \[
\tdet^{(i)} \;:=\; \min\Big\{\, t>\tvio \;\Big|\; e_{i,t} \ge e_{\text{th}} \Big\}\footnote{We have the convention that \(\tdet^{(i)}=+\infty\) if the set is empty.
}.
\]
The earliest network detection time is
\(
\tdet \;:=\; \min_{i\in\agentSet}\, \tdet^{(i)}.
\)
The set of agents that detect the contradiction within the window
\([\tvio,\ \tvio+\ndct]\) is
\[
\agentSetdet \;:=\; \big\{\, i\in\agentSet \ \big|\ \tdet^{(i)} \in [\tvio,\ \tvio+\ndct] \big\}.
\]

By Assumption~\ref{assump:contradiction_detection}, this set is nonempty\footnote{In practice, choosing \(L \le \ndct\) ensures \(\tdet - \tvio \le \ndct\) provided there exists at least one agent whose post-change residual exceeds the threshold at least $e_{\text{th}}$ times per window.}. After confirming a contradiction, at times $t \ge \tdet$, each agent $i \in \agentSetdet$ picks $\uEpsit=\textsc{explore}$ to explore the environment to collect evidence that discriminates among the hypotheses in \(\Phi\). During this exploration phase, each new observation is mapped to a pairwise evidence increment. Specifically, for each pair of hypotheses $\formula_\kpro,\formula_\lpro \in \Phi$, each agent $i \in \agentSetdet$ maintains a \emph{cumulative pairwise evidence} score $\sco_{i,t}^{\kpro,\lpro}$, which quantifies the evidence in favor of hypothesis $\formula_\kpro$ versus hypothesis $\formula_\lpro$, defined as:
\[
\scos_{i,t}^{\kpro,\lpro} = \sum_{t' = \tdet +1}^{t} \sco_{i,t'}^{\kpro,\lpro}, \qquad \kpro \neq \lpro,
\]
where the per-step score $\sco_{i,t'}^{\kpro,\lpro}$ is bounded and has a positive mean gap under the true hypothesis (e.g., log-likelihood ratio, discriminative margin, or consistency vote). 
\noindent
Based on these cumulative scores, agent \(i \in \agentSetdet\) selects its current best hypothesis $\formula_{\hat{\kpro}_i(t)}$ via a maximin rule:
\[
\formula_{\hat{\kpro}_i(t)} \;\in\; \argmax_{\kpro:\formula_\kpro\in\Phi}\ \m_{i,t}^{\kpro},
\quad
\text{where}\quad
\m_{i,t}^{\kpro} \;:=\; \min_{\lpro: \formula_{\lpro}\in\Phi\setminus \formula_{\kpro}}\ \scos_{i,t}^{\kpro,\lpro}.
\]
\noindent Note that any fixed tie–breaking rule may be used. With these maximin scores in place, exploration proceeds until the leading hypothesis is uniformly ahead by a predefined margin $\etaE$. Agent \(i\) stops collecting evidence at the first time the chosen hypothesis is
sufficiently separated from all competitors:
\[
\tau_i \;:=\; \min\Big\{\, t \ge \tdet^{(i)} \ \Big|\ \m^{\,\hat{\kpro}_i(t)}_{i,t} \;\ge\; \etaE \Big\}.
\]
The threshold \(\etaE\) is selected to attain a desired error bound, e.g.
\(\Pr\!\big\{\hat{\kpro}_i(\tau_i) \neq \kpro^\star\big\} \le \epsilon\),
where \(\kpro^\star\) is the index of the true hypothesis.
Once the separation condition is met, exploration terminates and the agent announces its current best hypothesis. At time $t=\tau_i$, agent $i$ \textsc{broadcasts} a signed message
\[
\msg_{i,\tau_i}\;=\;\bigl(\tau_i,\ \knows_i \varphi_{\hat{\kpro}_i(\tau_i)},\ \m^{\hat{\kpro}_i(\tau_i)}_{i,\tau_i}\bigr)
\]
to all other agents with time-to-live $\mathrm{TTL}=\lcomm$. If multiple agents in $\agentSetdet$ broadcast conflicting hypotheses, the network resolves the disagreement by selecting the hypothesis with the highest cumulative evidence among all broadcasts received up to time $t$ (with any fixed tie–breaking rule). Once the most plausible hypothesis is selected, the goal shifts from evidence gathering to belief alignment. Each agent $i \in \agentSet$ sets $\uEpsit=\textsc{revise}$ to revise its epistemic state $\Istate_{i,t}$ so that all currently accessible worlds satisfy \(\varphi_{\hat{\kpro}}\). We then define the epistemic recovery time $\trecE$ as the first $t$ such that all agents knows $\varphi_{\hat{\kpro}}$:
\[
\trecE \;:=\; \min\{\, t > \tvio \mid (\kripke,\run,t)\models  \everyK_\agentSet \varphi_{\hat{\kpro}} \,\}.
\]


\noindent Therefore,
\begin{align}
\Delta \trecE(\etaE) \;&=\; \trecE - \tvio \nonumber \\
\;&= \underbrace{(\tdet - \tvio)}_{\text{detection delay}}
+
\underbrace{(\taumax - \tdet)}_{\text{evidence aggregation}}
+
\underbrace{\lcomm}_{\text{message propagation}} \nonumber\\
\;&\le\; \ndct \;+\; (\taumax - \tdet) \;+\; \lcomm, \label{eq:boundEp}
\end{align}

\noindent where \(\taumax := \max_{i \in \agentSetdet} \tau_i\) is the time stamp of the last report. Note that the dependency of $\trecE$ and $\taumax$ on $\etaE$ is left implicit for ease of presentation. Figure~\ref{fig:epi_policy} depicts the proposed epistemic policy executed by each agent and the associated information flow.

\noindent\textbf{Meeting epistemic recoverability specifications.} To meet the epistemic recoverability specification, the evidence threshold $\etaE$ should be chosen so that
\begin{equation}
\label{eq:recEq}
\Delta \trecE(\etaE)=\trecE-\tvio \;\le\; \maxRE.
\end{equation}

\noindent Using the broadcast/commit rule, Equations~\eqref{eq:boundEp} and~\eqref{eq:recEq}, and Assumptions~\ref{assump:communication}-\ref{assump:contradiction_detection} a sufficient condition for recovering within at most $\maxRE$ steps is
\[
\tau_{\max}-\tdet \;\le\; \maxRE-\ndct-\lcomm,
\]
(deterministically or with high probability). This requires that $\maxRE \ge \ndct +\lcomm$. If this fails, the pair $(\maxRE,\etaE)$ is infeasible under the current sensing/communication regime.

\noindent\textbf{Meeting epistemic durability specifications.} With worst–case disturbance spacing $\drate$ (Assumption~\ref{assump:freq_dst} and Remark~\ref{rem:time_budget}), the achieved epistemic durability is lower bounded by
\[
\Delta\tdurE(\etaE)\ \ge\ \drate-\Delta\trecE(\etaE),
\]
provided the post-change hypothesis is correctly identified and remains valid until the next disturbance. Indeed, the threshold $\etaE$ trades off the probability of a correct commit,
\(p(\etaE):=\Pr\{\hat{\kpro}(\trecE)=\kpro^\star\}\),
against recovery time $\Delta\trecE(\etaE)$ and in turns the durability time $\Delta\tdurE(\etaE)$. A practical tuning rule is:
\begin{align}
\label{eq:obj_eta_epi}
\begin{aligned}
\max_{\etaE}\quad & p(\etaE)\,\bigl(\drate-\Delta\trecE(\etaE)\bigr) \\
\text{s.t.}\quad & 0 \le \Delta\trecE(\etaE)\ \le\ \min\{\maxRE,\ \drate-\minDE\},
\end{aligned}
\end{align}
so that the durability budget is nonnegative and at least meets $\minDE$. Intuitively: gather enough evidence to ensure correctness (so
the belief is likely to persist) but not so much that recovery
consumes the inter-disturbance interval. In practice, pick $\etaE$ to meet a target error while satisfying recoverability $(\maxRE)$ and durability $(\minDE)$ for nominal parameters (e.g., $\ndct,\lcomm$). In turn, this reduces to a one–dimensional line search over a feasible interval $[\etaE_{\min},\etaE_{\max}]$; if the interval is empty, the specification is \emph{infeasible}.

\addb{\subsection{Achieving Action Resilience}

After epistemic recovery at time $\trecE$, agents must re-align their external
policies to the newly identified condition $\varphi_{\hat{\kpro}}$ and sustain
optimal performance. Thus, once beliefs are revised, each agent $i$ runs a
best-policy identification loop over its finite candidate set
$\Pi_i^{\mathrm{ext}}$. At times $t=\trecE,\trecE+1,\ldots$, agent $i$:

\begin{enumerate}
  \item \textbf{Roll out.} Execute a short rollout of a candidate policy
        $\pi_i\in\Pi_i^{\mathrm{ext}}$ and record its task cost.

  \item \textbf{Update.} For each $\pi_i\in\Pi_i^{\mathrm{ext}}$, maintain the
        empirical mean cost $\hat J_{i,t}(\pi_i)$ and a confidence radius
        $\conf_{i,t}(\pi_i)$ calibrated to confidence level $\etaA$; the radius
        shrinks as more rollouts are collected. Formally, $\hat J_{i,t}(\pi_i)$
        is an estimator of the true performance value $J_i(\pi_i; t)$ introduced
        in Section~\ref{sec:metrics}, and the confidence interval
        $[\hat J_{i,t}(\pi_i) - \conf_{i,t}(\pi_i),\,
          \hat J_{i,t}(\pi_i) + \conf_{i,t}(\pi_i)]$
        is constructed so that, with high probability, it contains
        $J_i(\pi_i; t)$ for all $\pi_i\in\Pi_i^{\mathrm{ext}}$.

  \item \textbf{Select/stop.} Explore policies with the smallest lower confidence
        bound and stop as soon as one policy is statistically separated from all
        others; adopt a policy $\pi_i^{*}$ if
        \[
          \hat J_{i,t}(\pi_i^{*}) + \conf_{i,t}(\pi_i^{*})
          \;<\;
          \min_{\pi_i\in\Pi_i^{\mathrm{ext}}\setminus\{\pi_i^{*}\}}
          \big[\hat J_{i,t}(\pi_i) - \conf_{i,t}(\pi_i)\big],
        \]
        i.e., the worst plausible cost of $\pi_i^{*}$ is below the best plausible cost of every other policy. Under this condition, $\pi_i^{*}$ is, with
        high probability, an element of the optimal set $\Pi_{i,t}^{\mathrm{opt}}$
        associated with $J_i(\cdot; t)$.
\end{enumerate}
Having selected such a policy $\pi_i^{*}$, we now formalize when the system has
completed \emph{action} recovery. The atomic proposition $\pi_i^{\mathrm{opt}}$ is defined
to hold exactly when the agent’s current external policy belongs to the set $\Pi_{i,t}^{\mathrm{opt}}$:
\[
(\kripke,\run,t)\models \pi_i^{\mathrm{opt}}
\qquad\Longleftrightarrow\qquad
\pi^{\mathrm{ext}}_{i,t}\in \Pi_{i,t}^{\mathrm{opt}}.
\]
The group-level predicate is the conjunction
\[
\pi^{\mathrm{opt}} \equiv \bigwedge_{i\in\agentSet} \pi_i^{\mathrm{opt}},
\]
which holds precisely when every agent executes one of its individually optimal
external policies. No joint optimality or equilibrium concept is assumed;
optimality is evaluated individually for each agent.

The \emph{action recovery time} is therefore
\[
\trecA := \min\{\, t \ge \trecE \mid (\kripke,\run,t)\models \pi^{\mathrm{opt}} \,\},
\]
and the recovery duration is
\[
\Delta\trecA(\etaE,\etaA) = \trecA - \trecE,
\]
where the dependence on $\etaA$ is omitted for readability.
}

\noindent\textbf{Meeting action recoverability specifications.} To meet the action recoverability specification, the confidence level $\etaA$ should be chosen so that $\Delta \trecA(\etaA) \;\le\; \maxRA$. If the best policy cannot be statistically separated before the deadline \(\maxRA\) in accordance with $\etaE$ and $\etaA$, then the
specification \((\maxRA,\etaE,\etaA)\) is \emph{infeasible} under the current candidate set
\(\Pi_i^{\mathrm{ext}}\) and noise/rollout budget. In that case one must adjust \(\maxRA\), \(\etaE\), and/or \(\etaA\), or inject stronger priors.

\noindent\textbf{Meeting action durability specifications.} With worst–case inter–disturbance spacing \(\drate\) (Assumption~\ref{assump:freq_dst} and Remark~\ref{rem:time_budget}), the \emph{achieved} action durability satisfies
\[
\Delta\tdurA(\etaE,\etaA) \;\ge\; \drate \;-\; \Delta\trecE(\etaE) \;-\; \Delta\trecA(\etaE,\etaA),
\]
provided that the post–change hypothesis and the optimal policy are correctly identified. Using \(\Delta\trecE(\etaE)\le \maxRE\) gives 
\[
\Delta\tdurA(\etaE,\etaA) \;\ge\; \drate \;-\; \maxRE \;-\; \Delta\trecA(\etaE,\etaA).
\]
The confidence level \(\etaA\) trades off between the probability of correct policy identification $p(\etaA)$, against the recovery time \(\Delta\trecA(\etaE,\etaA)\) and the durability time \(\Delta\tdurA(\etaE,\etaA)\). A practical tuning rule is
\begin{align}
\label{eq:obj_eta_act}
\begin{aligned}
\max_{\etaA}\quad & p(\etaA)\,\Big(\drate - \maxRE - \Delta\trecA(\etaA)\Big) \\
\text{s.t.}\quad  & 0 \le \Delta\trecA(\etaE,\etaA)\le \min\{\,\maxRA, \,\drate-\maxRE - \minDA\}.
\end{aligned}
\end{align}
so the guaranteed durability budget is nonnegative and at least meets \(\minDA\).
Intuitively: \emph{explore policies enough to ensure correctness (so the chosen policy is likely to persist)
 but not so much that recovery consumes the inter-disturbance interval.} In practice, \(\etaA\) is chosen to meet a target error, and and satisfy $(\maxRA,\minDA)$ under nominal sensing parameters. This reduces to a one–dimensional line search over a feasible interval
\([\etaA_{\min},\,\etaA_{\max}]\); if the interval is empty, the specification is \emph{infeasible}.

\addb{
\noindent \textbf{Computational complexity.}
The proposed architecture separates epistemic and action layers, and the
computational overhead can be assessed per agent.
On the epistemic side, each agent maintains a finite hypothesis set~$\Phi$
and the associated cumulative pairwise evidence scores.
Updating the residual $\Upsilon(\Ystate_{i,t},\Ystateh_{i,t})$ requires
$O\!(\dim \Yspace)$ computations, and maintaining the sliding-window
threshold–exceedance counter requires $O\!(1)$ computation per step.
During the \textsc{explore} phase, updating all pairwise evidence scores and
computing the maximin statistic $\m_{i,t}^{\kpro}$ require
$O\!(|\Phi|^{2})$ computations per update step. On the action side, each agent $i$ performs best-policy identification over its
finite candidate set~$\Pi_i^{\mathrm{ext}}$.
Updating empirical means and confidence radii, together with evaluating the
stopping condition, requires $O\!(|\Pi_i^{\mathrm{ext}}|)$ computations per update step. Overall, the worst-case per-agent per-step computational complexity of the framework is
\[
  O\!\big( \dim \Yspace \;+\; |\Phi|^{2} \;+\; |\Pi_i^{\mathrm{ext}}| \big),
\]
that is, linear in the observation dimension and in the number of candidate
external policies, and quadratic in the number of epistemic hypotheses. The framework is compatible with resource-constrained deployments because
its dominant computational terms depend on the cardinalities of the
hypothesis set~$\Phi$, the external policy set~$\Pi_i^{\mathrm{ext}}$, and the
dimension of the observation space~$\Yspace$. These quantities can be made small
by design (e.g., coarse-graining hypotheses, restricting candidate policies,
or reducing observation representations), thereby allowing system designers to
trade model resolution for computational footprint.
}

\section{Formal Verification of Resilience Properties}
\label{sec:analytical}
We evaluate formulas along an actual run \(\run\) using \((\kripke,\run,t)\models\cdot\) as defined in Section~\ref{sec:kripke}.
A disturbance at time \(\timeV\) is a change in the environment such that a pre-change fact \(\varphi_1\)
holds for all \(t<\timeV\) along \(\run\), and the post-change fact \(\varphi_2\) holds for all \(t\ge \timeV\)
along \(\run\); agents’ relations \(\rel_{i,t}\) evolve with \(t\).

\begin{theorem}[\textbf{Soundness of the Resilience Specifications}]
\label{thm:soundness}
If \((\mathcal{M}_{\timeV},\run,\timeV)\models \resE \land \resA\), then along \(\run\):
\begin{enumerate}
  \item The epistemic recovery time \(\Delta\trecE\) defined in~\eqref{trecE} satisfies \(\Delta\trecE\le \maxRE\), and the epistemic durability time defined in~\eqref{tdurE} satisfies \(\Delta\tdurE\ge \minDE\).
  \item The action recovery time \(\Delta\trecA\) defined in~\eqref{trecA} satisfies \(\Delta\trecA\le \maxRA\), and the action durability time defined in~\eqref{tdurA} satisfies \(\Delta\tdurA\ge \minDA\).
\end{enumerate}
\end{theorem}
\begin{proof}
By the semantics of bounded \(\until\) and \(\glob\),
\((\mathcal{M}_{\timeV},\run,\timeV)\models \resE\) implies there exists \(\trecE\in\{\timeV,\ldots,\timeV+\maxRE\}\)
such that \((\mathcal{M}_{\trecE},\run,\trecE)\models \everyK_\agentSet \varphi_2\) and
\((\mathcal{M}_{t'},\run,t')\models \everyK_\agentSet \varphi_2\) for all
\(t'\in\{\trecE,\ldots,\trecE+\minDE-1\}\). Hence \(\Delta\trecE\le \maxRE\) and \(\Delta\tdurE\ge \minDE\). Similarly, \((\mathcal{M}_{\timeV},\run,\timeV)\models \resA\) yields \(\trecA\in\{\timeV,\ldots,\timeV+\maxRA\}\)
such that \((\mathcal{M}_{\trecA},\run,\trecA)\models \pi^{\mathrm{opt}}\) and
\((\mathcal{M}_{t'},\run,t')\models \pi^{\mathrm{opt}}\) for all
\(t'\in\{\trecA,\ldots,\trecA+\minDA-1\}\). Thus \(\Delta\trecA\le \maxRA\) and \(\Delta\tdurA\ge \minDA\).
\end{proof}

\begin{corollary}[\textbf{Bounded-Horizon Completeness (Finite Counterexample)}]
\label{cor:completeness}
If \((\mathcal{M}_{\timeV},\run,\timeV)\not\models \resE \land \resA\), then there exists a finite time
\(t \in [\timeV,\ \timeV+B]\) witnessing a violation, where
\[
B \;:=\; \max\{\, \maxRE+\minDE,\ \maxRA+\minDA \,\}.
\]
Equivalently, to refute \(\resE \land \resA\) it suffices to examine at most \(B\) steps
after the disturbance.
\end{corollary}
\begin{proof}
Consider the epistemic part \(\resE \equiv (\neg\,\everyK_\agentSet\varphi_2)\ \until_{[0,\maxRE]}\ 
(\glob_{[0,\minDE)}\,\everyK_\agentSet\varphi_2)\).
If \(\resE\) fails at \(\timeV\), then either:
(i) no \(\trecE \le \timeV+\maxRE\) exists with \((\mathcal{M}_{\trecE},\run,\trecE)\models \everyK_\agentSet\varphi_2\),
so the failure is witnessed within \(\maxRE\) steps; or
(ii) some such \(\trecE\) exists but \(\glob_{[0,\minDE)}\everyK_\agentSet\varphi_2\) fails, so there is
an \(t' \in \{\trecE,\ldots,\trecE+\minDE-1\}\) with \((\mathcal{M}_{t'},\run,t')\models \neg\,\everyK_\agentSet\varphi_2\),
hence a witness within \(\maxRE+\minDE\) steps.
An identical case split for \(\resA\) yields a witness within \(\maxRA+\minDA\) steps.
Taking the maximum of the two bounds gives \(B\).
\end{proof}
Formal verification entails checking, at design time or at runtime, whether a multi-agent system satisfies the resilience properties under the stated assumptions. In particular, by Theorem~\ref{thm:soundness} and Corollary~\ref{cor:completeness}, verification reduces to bounded reasoning: it suffices to examine at most \(B=\max\{\maxRE+\minDE,\ \maxRA+\minDA\}\) steps after a disturbance. \addb{Since the verification is confined to a finite horizon, the resilience
specifications become compatible with standard bounded model-checking approaches,
such as SAT/SMT-based unrolling, explicit-state exploration over a window of
length \(B\), and lightweight runtime monitors that track only the last \(B\)
steps. Thus, the formal results established in this section enable 
computationally tractable certification of recovery protocols, knowledge-update rules, and coordination strategies before deployment, as well as efficient runtime monitoring in dynamic environments.}

\section{Numerical Results and Analysis}
\label{sec:numerical}
To show the effectiveness of our proposed framework, we study a distributed multi-agent decision-making under abrupt network stressors usecase. In particular, we study a group of networked agents $\agentSet$ repeatedly selecting among $A = |\UExt|$ actions while the environment undergoes changes (stressors). All agents are assumed to share the same set of external actions, i.e., $\UExt_1 = \cdots = \UExt_\agentsize = \UExt$. Each action, denoted $\arm$ for ease of presentation where $\arm = \uExt \in \UExt$, yields a Gaussian reward with variance $\var^2$ and a mean $\mean_\arm(t)$ that might change due to environmental stressors. This models operational regimes where external events suddenly alter performance; e.g., the activation of a strong interferer, mobility-driven shadowing, or the appearance/disappearance of a hotspot. \addb{The group’s goal is to maximize total reward (equivalently, minimize cumulative
regret) over the horizon~$\timeD$. In the setting considered here, this
is equivalent to each agent $i$ maximizing its own expected reward $J_i$, since
all agents face the same action set and reward structure.
}
This abstraction, commonly referred to in the literature as the multi-agent multi-armed bandit, encapsulates several canonical tasks in networked communication systems where fast, decentralized adaptation is required. In this case, external actions are also referred to as arms.

Classical solutions to the multi-arm bandit problem typically rely on statistical algorithms such as Upper Confidence Bound (UCB)~\cite{UCB_MAMAB}, Thompson Sampling~\cite{Thompson_MAMAB}, or epsilon-greedy strategies~\cite{greedy_MAMAB}. These methods balance exploration and exploitation using statistical estimates of action rewards, confidence bounds, or posterior distributions. However, these strategies inherently assume that reward deviations are purely stochastic and unbiased. This assumption fails under adversarial or non-stationary conditions. To handle non-stationarity, one standard strategy is passive adaptation via fading memory. This approach, exemplified by Discounted-UCB, which down-weight stale data to track piecewise-constant drifts~\cite{Discount_UCB}. 

In this work, we employ UCB as the underlying algorithm for all baselines and for our approach. In particular, for each action $a$, every agent $i$ maintains two running statistics (i) \textit{count}: the number of times agent $i$ has pulled action $a$ up to current time step, and (ii) \textit{sum}: the cumulative sum of rewards agent $i$ has observed from action $a$ up to current time step. These two quantities are the sufficient statistics needed to form the empirical mean of action $a$; $\mean_\arm(t)$. Our implementation uses the Gaussian-UCB algorithm found in~\cite[Chapter~7]{UCB_Eq}. In the cooperative baseline, agents do not share raw reward histories; they exchange only the per-action counts and sums with their neighbors and average them via a linear consensus step. After mixing, each agent scales the per-agent network averages by the group size $\agentsize$ to approximate global counts and sums, forming a globalized empirical mean for use inside the UCB decision rule. 

\textbf{Baselines:} We simulate two standard schemes. In the \textbf{Independent Discounted-UCB}, each agent runs a discounted-UCB policy using only its own observations, where an exponential forgetting rate enables adaptation to non-stationary rewards. In the \textbf{Cooperative Discounted-UCB}, agents are connected through an undirected communication graph and apply a consensus-based discounted-UCB, sharing sufficient statistics with neighbors to accelerate adaptation.

\textbf{Proposed approach:} We consider two variants of our epistemic (Kripke) framework. In the first variant, as \textbf{Light-Coop Kripke}, agents follow independent UCB dynamics and communicate only when announcing a detected change in their knowledge or belief state. In the second variant, referred to as \textbf{Cooperative Kripke}, agents operate under cooperative UCB from the outset; they not only exchange belief-change messages but also jointly perform cooperative recovery, pooling both evidence and action statistics to identify and realign to the new correct hypothesis. 

\textbf{Simulation Setup:} We simulate a network of $\agentsize = 12$ agents operating synchronously over a time horizon of $\timeD = 2500$ steps. The action set consists of $|\UExt| = 16$ actions (arms), each arm reward's true mean is randomly picked from interval $[0.1,1.2]$. Agents are connected through an undirected communication graph instantiated as a ring topology. The environment follows a piecewise-stationary structure: each arm’s true mean remains constant within a segment and changes abruptly at a predefined change-point (unknown to agents) representing an external stressor. In the main experiment, the stressor occurs at $t=1400$, causing a reordering of the optimal arms. For cooperative algorithms, agents exchange local statistics with neighbors using Metropolis–Hastings weights and perform one consensus round per decision epoch. All methods are initialized with identical random seeds to ensure that they experience the same stochastic realizations of the environment across trials. Each configuration is repeated over $10$ ndependent trials, and reported performance curves show the sample mean together with $95\%$ confidence intervals computed over these trials. To improve visualization, the curves are lightly smoothed using a moving average with a window length of $50$ time steps. Moreover, to examine resilience under different noise conditions, we conduct simulations for two noise levels: $\sigma \in \{0.5, 1.0\}$. This allows us to compare how independent and cooperative exploration balance under varying uncertainty. 

\addb{
The resilience bounds $(\maxRE,\minDE,\maxRA,\minDA)$ in our experiments are not algorithmic hyperparameters but designer-specified requirements on recovery and durability (see Section~\ref{sec:metrics} and Section~\ref{sec:algorithms}). Since the simulation environment contains only a \emph{single} stressor, these specifications are chosen so that resilience remains feasible within the fixed horizon $T$. In contrast, the detection thresholds $(\varepsilon_1, L, e_{\text{th}})$ govern when contradictory evidence is considered persistent rather than transient. These are selected from pre-change residual statistics (empirically ensuring that the false-alarm probability remains negligible). Finally, the epistemic evidence threshold $\etaE$ and action confidence threshold $\etaA$ are tuned using the one-dimensional rules in Equations~\eqref{eq:obj_eta_epi} and~\eqref{eq:obj_eta_act}, respectively. 
}

We model the bandit environment within our epistemic logic as follows. The set of
possible worlds is
\(
W=\{w_1,\ldots,w_H\},
\)
where each world \(w_h\) corresponds to a full vector of arm means
\(\mu^{(h)} \in \mathbb{R}^A\). At any given time, exactly one world is the
true environment condition.

We introduce atomic propositions evaluated world-by-world:

\begin{itemize}
  \item \textbf{Per-arm mean atoms:}
  \[
    p_{a,h}: \;\mu_a = \mu^{(h)}_a
    \qquad (a=1,\dots,A;\;h=1,\dots,H),
  \]
  \item \textbf{Actionable atoms:}
  \[
    b_a:\; a \;\text{is an optimal arm},
    \quad\text{ i.e., }\quad
    a \in \arg\max_{a'} \mu_{a'}.
  \]
\end{itemize}

The \emph{world-identity} formula
\[
\varphi_h \;=\; \bigwedge_{a=1}^{A}\, p_{a,h}
\]
characterizes world \(w_h\). Its actionable consequence is the formula
\[
\psi_h \;=\; \bigvee_{a\in\arg\max_j \mu^{(h)}_j} b_a,
\]
which states that the selected arm is optimal in world \(w_h\).
Log-likelihood ratios (LLRs) provide the cumulative pairwise evidence scores used
to discriminate among \(\{\varphi_h\}_{h=1}^H\).

\addb{In our simulations, the environment remains in a single world
$w_{h}$ for a contiguous time interval, and switches deterministically
to another world only at predefined disturbance (stressor) times. We instantiate the abstract optimal-action proposition for agent \(i\) as
\[
(\mathcal{M}_t,\run,t)\models \pi_i^{\mathrm{opt}}
\quad\Longleftrightarrow\quad
a_{i,t} \in \arg\max_a \mu_a(t),
\]
i.e., agent \(i\)’s chosen arm is optimal under the true world at time \(t\). Here, the external policy \(\pi_{i,t}^{\mathrm{ext}}\) refers to the behavior that
directly determines the arm selected by agent \(i\) at time \(t\), so that
\(a_{i,t} = \pi_{i,t}^{\mathrm{ext}}\).
Thus, \(\pi_i^{\mathrm{opt}}\) holds exactly when the agent’s external policy \(\pi_{i,t}^{\mathrm{ext}}\) selects an arm with the highest mean reward under the current environment condition. In this setting, the set $\Pi_{i,t}^{\mathrm{opt}}$ consists of all external policies that would be optimal under the current environment at time $t$. Therefore, the expected reward of agent $i$, \(J_i\), is therefore determined by its external
policy together with the time-varying sequence of worlds
\(\bigl(w_{h(t)}\bigr)_{t\ge 0}\). 
} Finally, performance is assessed through standard metrics: fraction of agents selecting the optimal action, total and cumulative rewards, and cumulative regret over time. These shared settings provide a uniform benchmark for contrasting the independent, cooperative, and epistemic (Kripke-based) learning schemes in terms of adaptability, recovery speed, and overall performance under abrupt environmental shifts. 


\begin{figure*}[t]
  \centering

  \subfloat[Total reward over time.\label{fig:sigma05:total_reward}]{
    \includegraphics[width=0.48\textwidth]{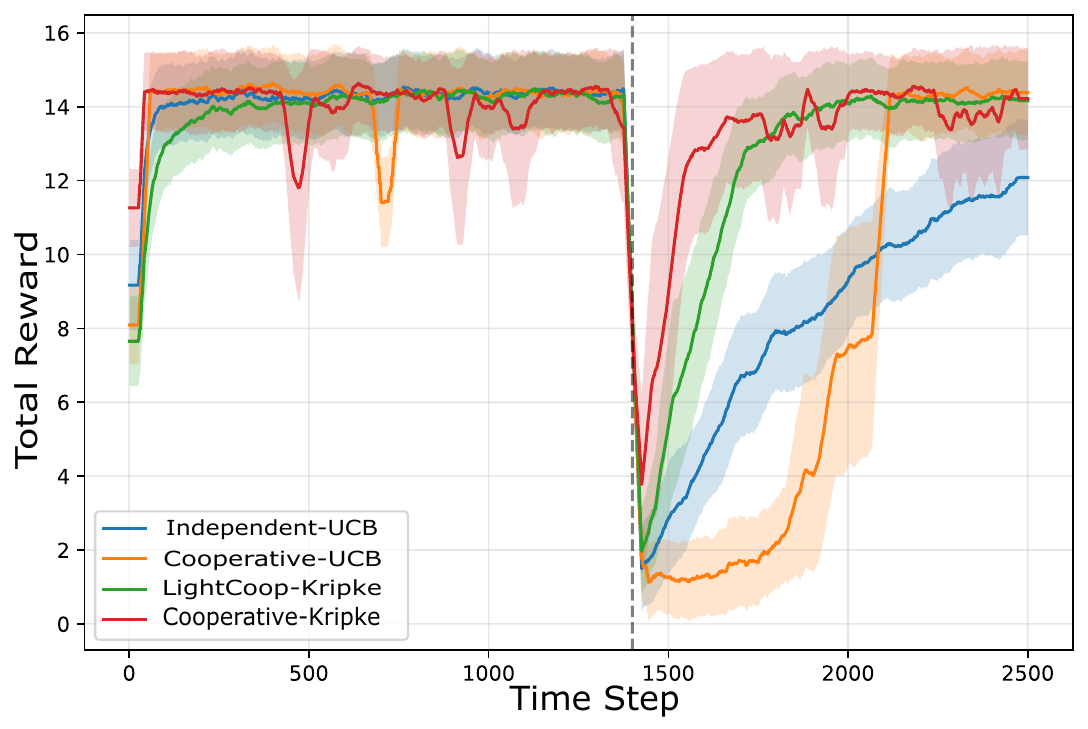}
  }\hfill
  \subfloat[Cumulative total reward.\label{fig:sigma05:cum_reward}]{
    \includegraphics[width=0.48\textwidth]{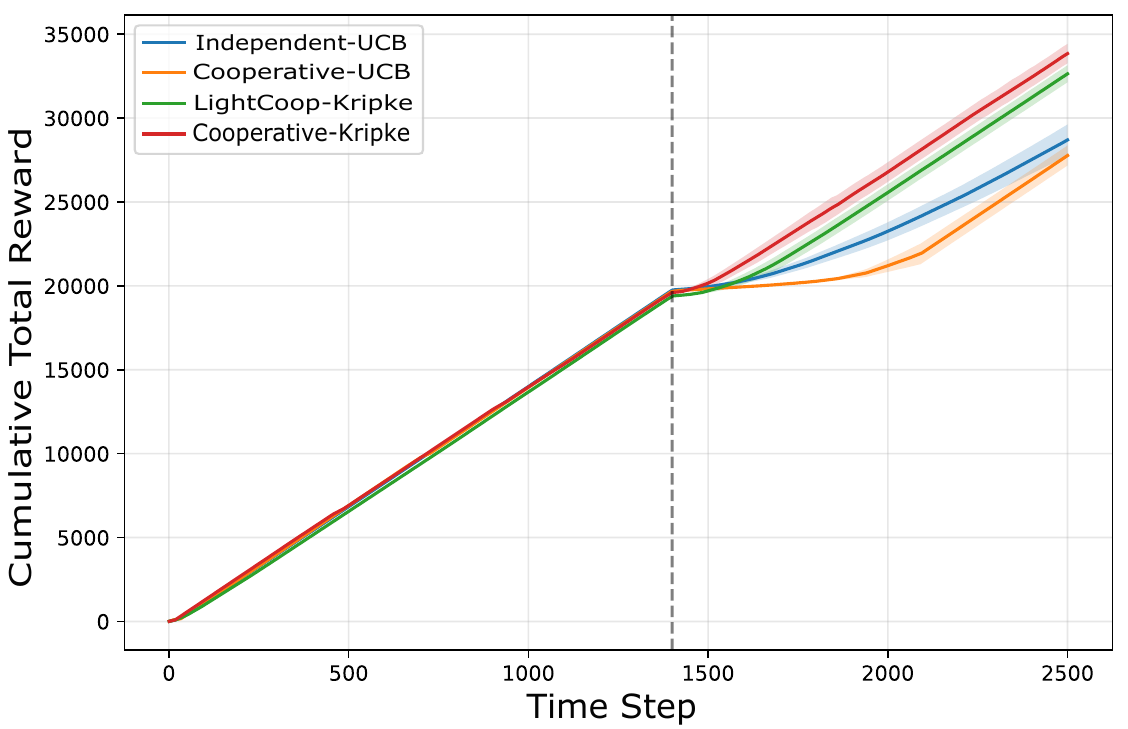}
  }

  \vspace{0.6em}

  \subfloat[Cumulative regret over time.\label{fig:sigma05:cum_regret}]{
    \includegraphics[width=0.48\textwidth]{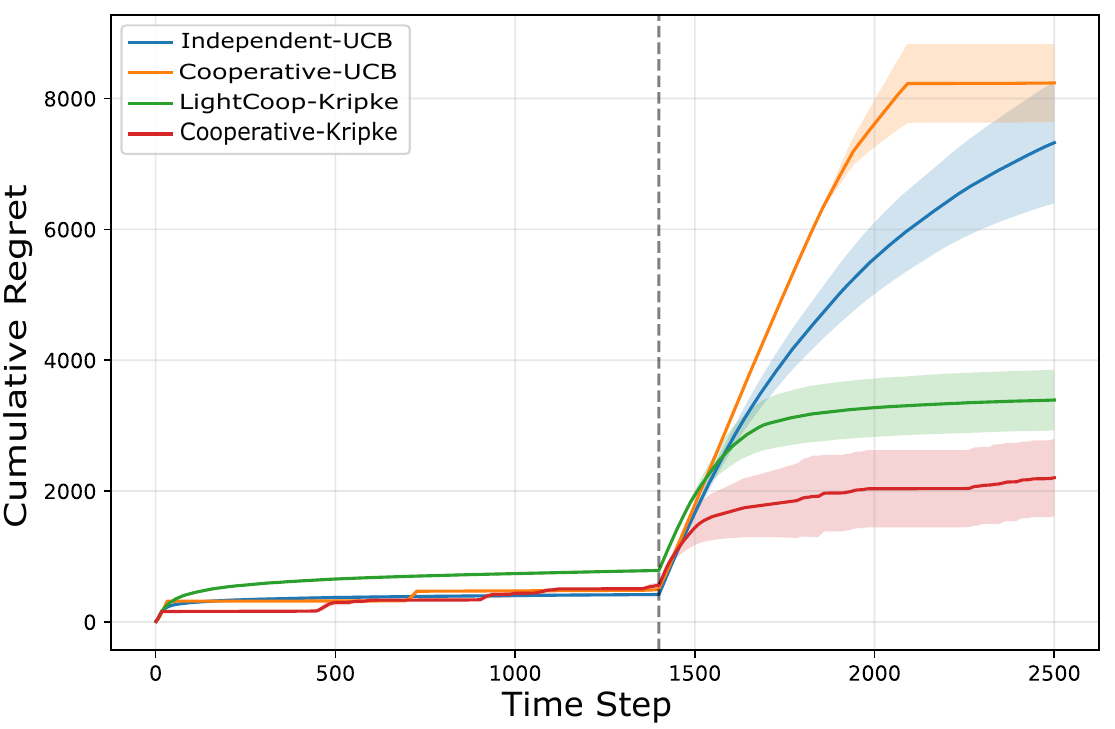}
  }\hfill
  \subfloat[Fraction of agents selecting the optimal action.\label{fig:sigma05:fract_optimal}]{
    \includegraphics[width=0.48\textwidth]{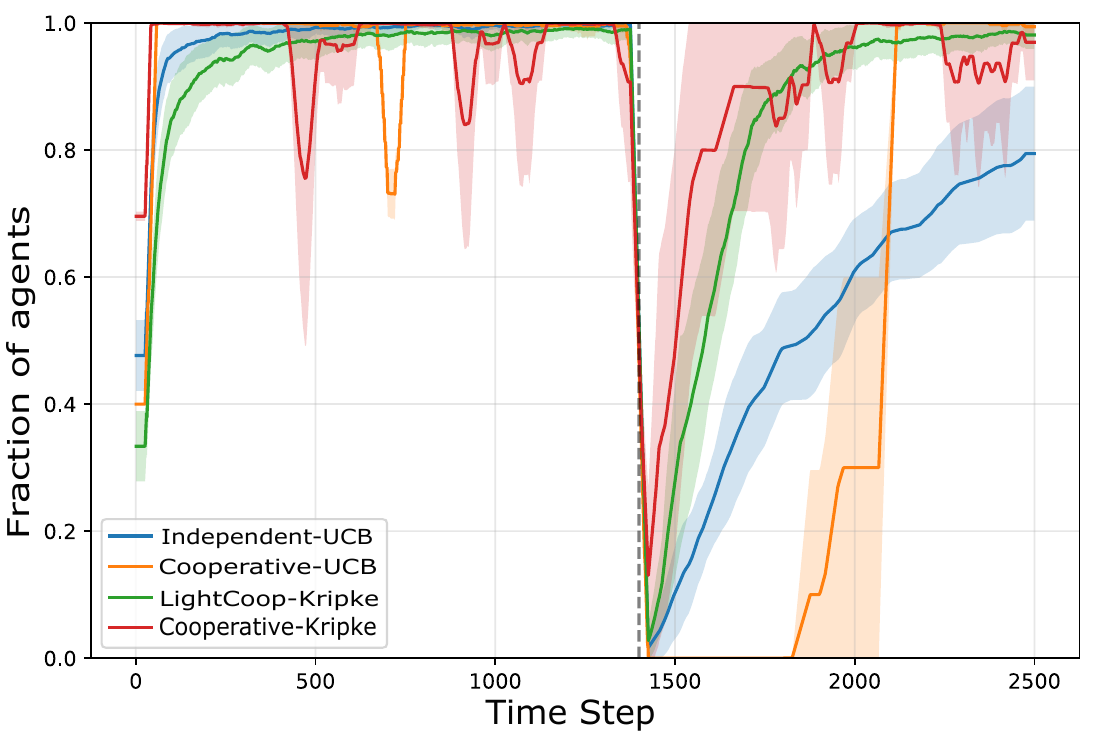}
  }

  \caption{Performance under noise level ($\sigma=0.5$). Each panel reports the mean across 10 trials with 95\% confidence intervals; the dashed vertical line marks the change-point at $t=1400$. For discounted-UCB baselines, the forgetting rate is set to $0.998$.}
  \label{fig:sigma05:main}
\end{figure*}

\begin{figure*}[t]
  \centering

  \subfloat[Total reward over time.\label{fig:sigma100:total_reward}]{
    \includegraphics[width=0.48\textwidth]{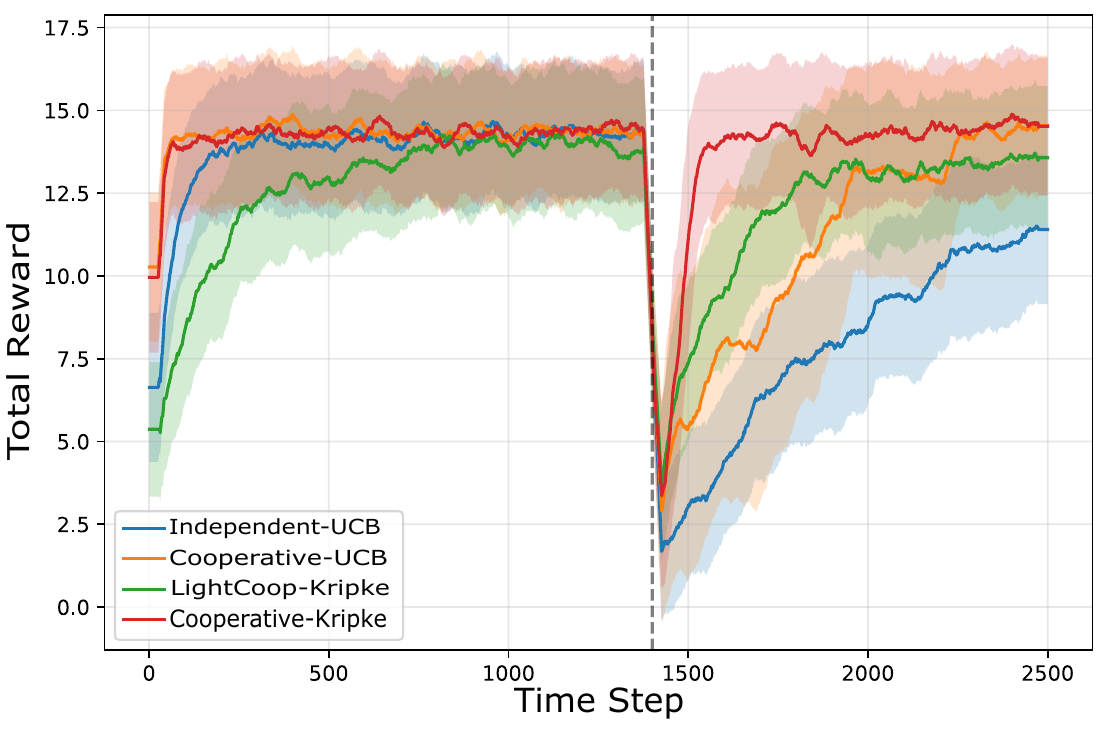}
  }\hfill
  \subfloat[Cumulative total reward.\label{fig:sigma100:cum_reward}]{
    \includegraphics[width=0.48\textwidth]{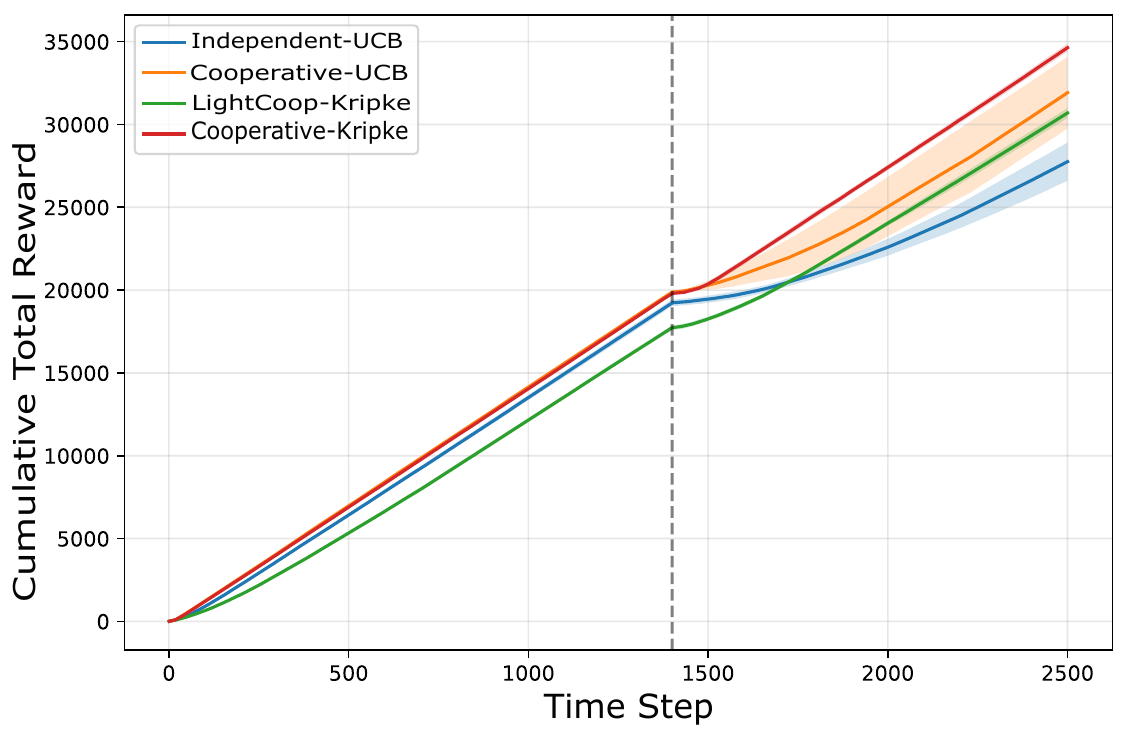}
  }

  \vspace{0.6em}

  \subfloat[Cumulative regret over time.\label{fig:sigma100:cum_regret}]{
    \includegraphics[width=0.48\textwidth]{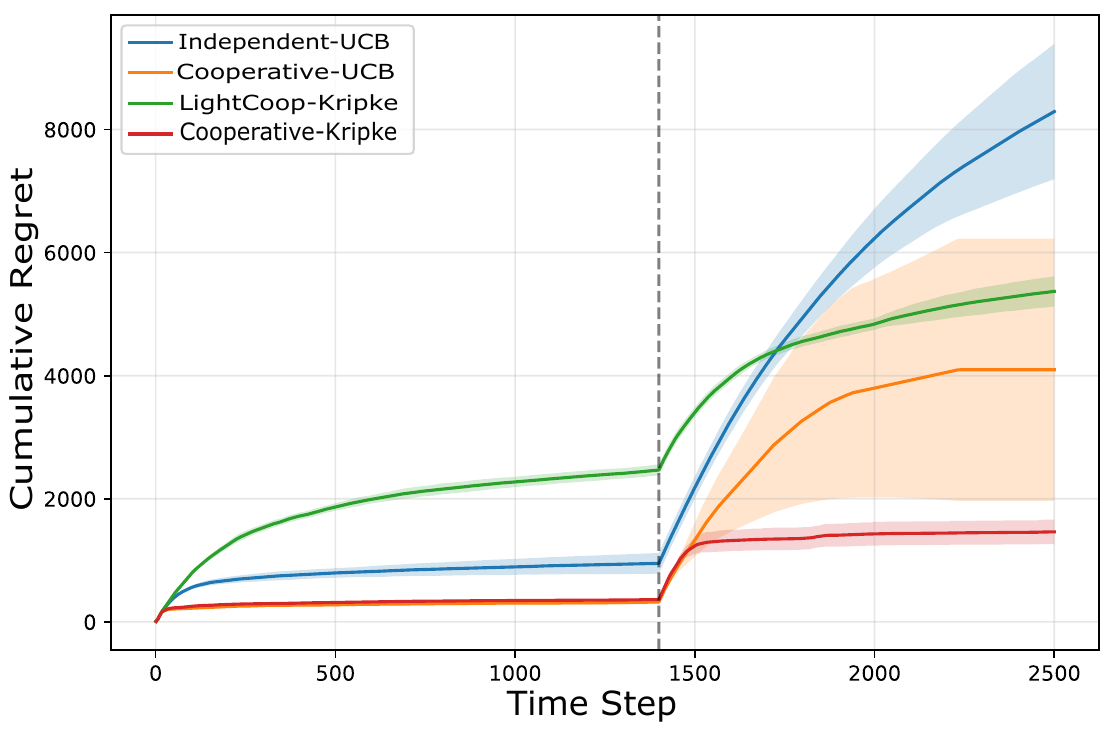}
  }\hfill
  \subfloat[Fraction of agents selecting the optimal action.\label{fig:sigma100:fract_optimal}]{
    \includegraphics[width=0.48\textwidth]{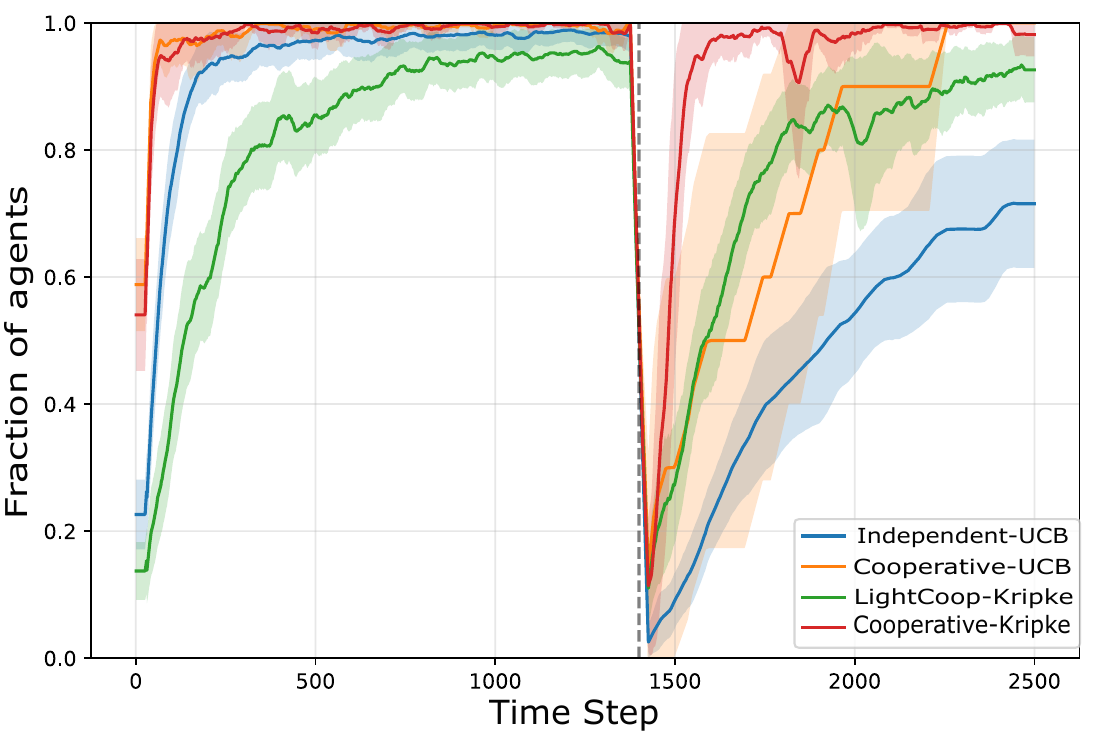}
  }

  \caption{Performance under noise level ($\sigma=1.0$). Each panel reports the mean across 10 trials with 95\% confidence intervals; the dashed vertical line marks the change-point at $t=1400$. For discounted-UCB baselines, the forgetting rate is set to $0.998$.}
  \label{fig:sigma10:main}
\end{figure*}

\begin{figure*}[t]
  \centering

  \subfloat[Total reward over time.\label{fig:add:total_reward}]{
    \includegraphics[width=0.48\textwidth]{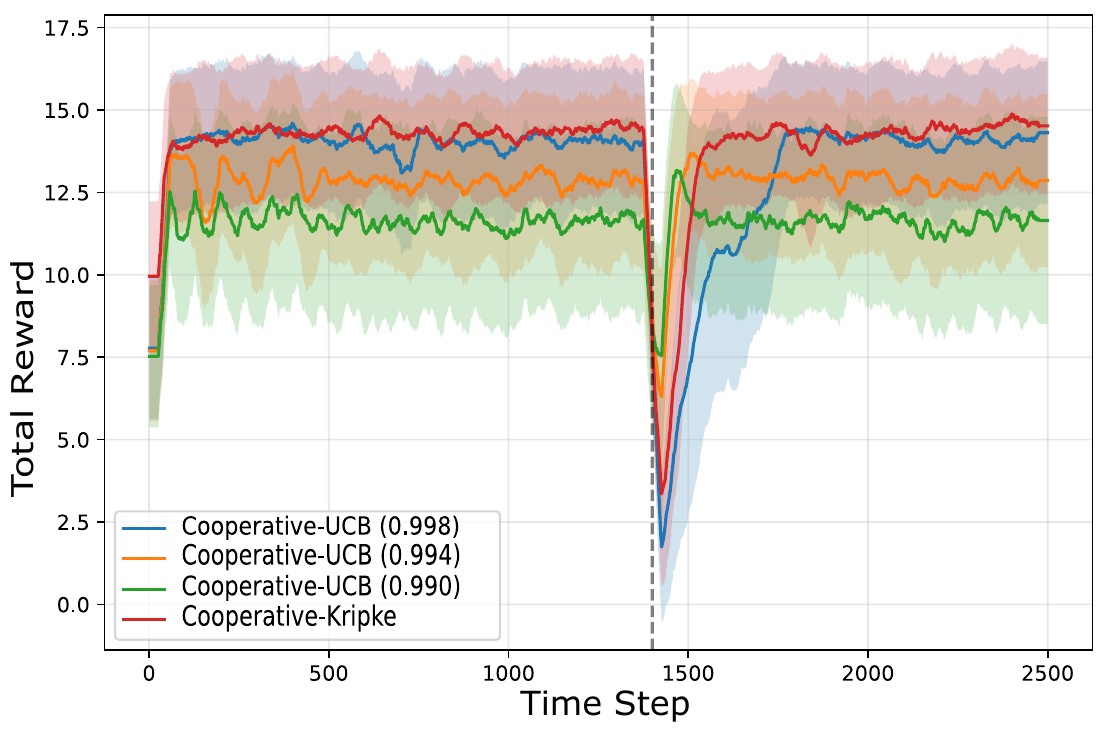}
  }\hfill
  \subfloat[Cumulative total reward.\label{fig:add:cum_reward}]{
    \includegraphics[width=0.48\textwidth]{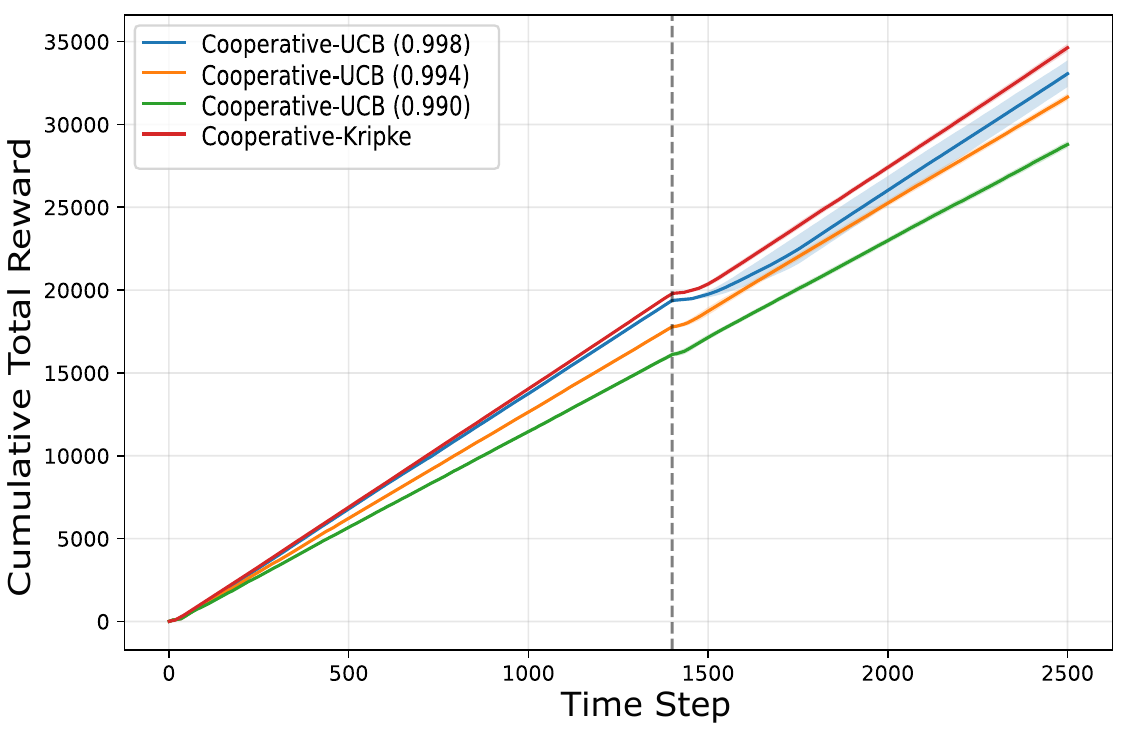}
  }

  \vspace{0.6em}

  \subfloat[Cumulative regret over time.\label{fig:add:cum_regret}]{
    \includegraphics[width=0.48\textwidth]{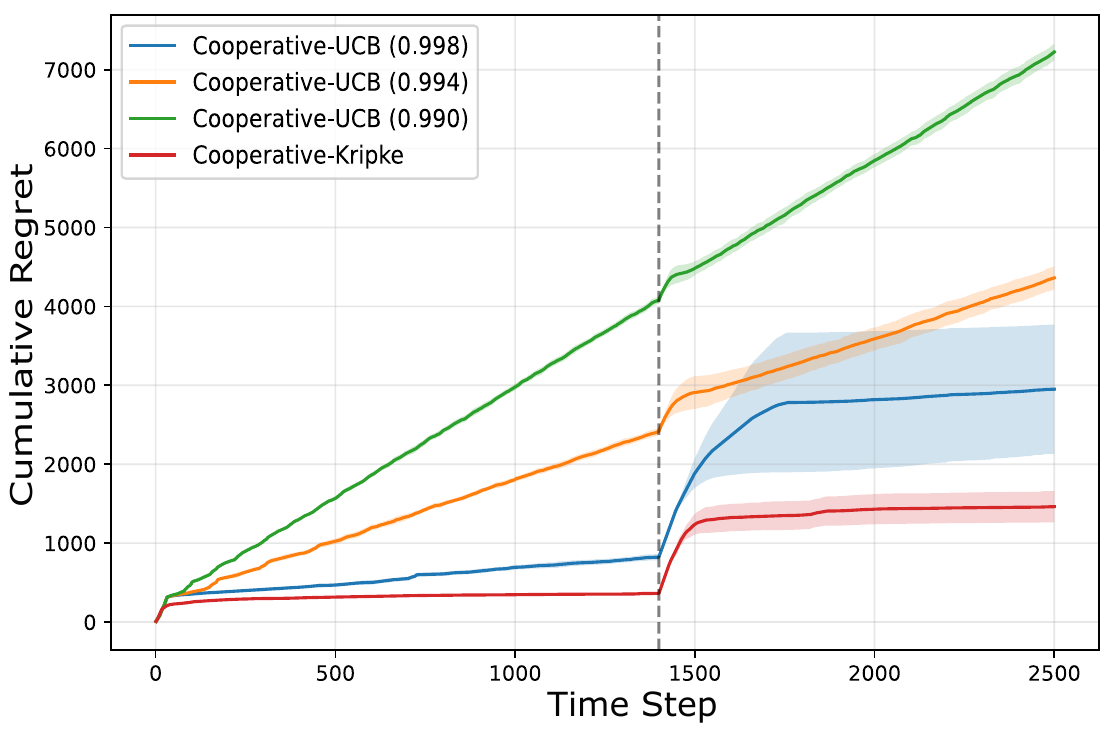}
  }\hfill
  \subfloat[Fraction of agents selecting the optimal action.\label{fig:add:fract_optimal}]{
    \includegraphics[width=0.48\textwidth]{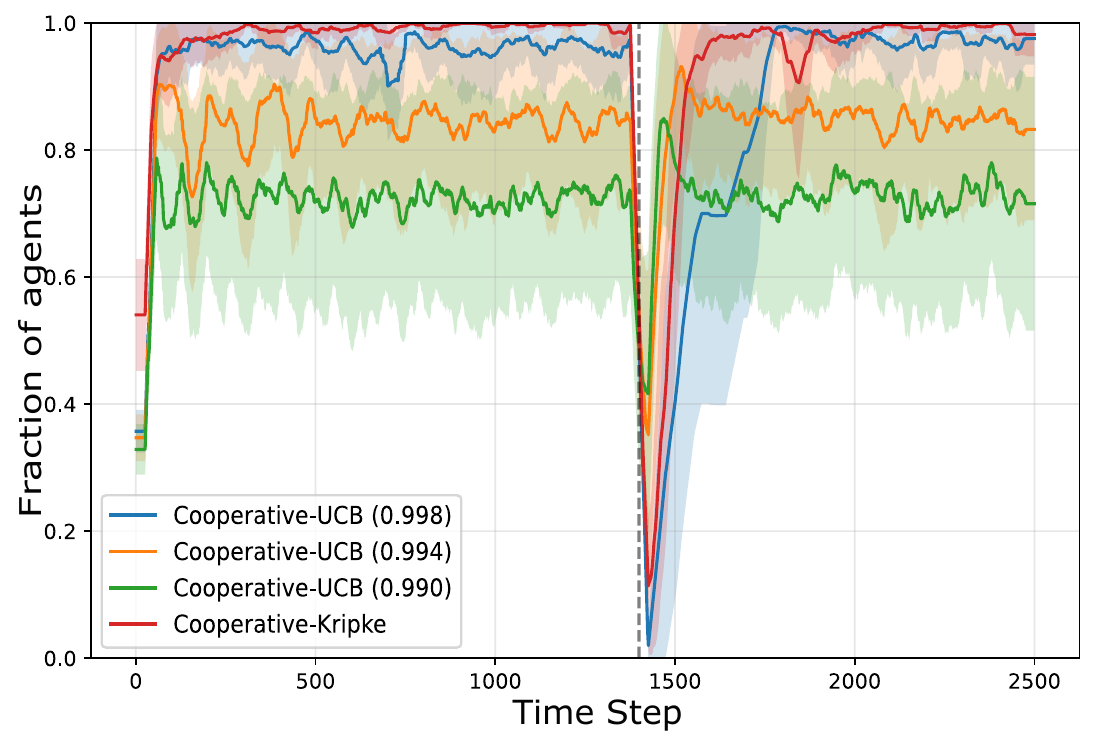}
  }

  \caption{Performance under noise level ($\sigma=1.0$). Each panel reports the mean across 10 trials with 95\% confidence intervals; the dashed vertical line marks the change-point at $t=1400$.}
  \label{fig:sigma10ADD:main}
\end{figure*}

\textbf{Moderate noise level}. Figure~\ref{fig:sigma05:main} summarizes the performance of all learning modes; Independent Discounted-UCB, Cooperative Discounted-UCB, and the proposed Independent-Kripke and Cooperative-Kripke frameworks when $\sigma = 0.5$. For simplicity, the term \textit{discounted} is omitted from the figure legends, although it applies to Independent and Cooperative UCB baselines. For this setup, the parameters for both Kripke learning mode is $\varepsilon_1 = 1.4$, $e_{\text{th}} = 12$, $L=30$, $\etaE =13$, $\maxRE = 550$, $\minDE = 600$, $\maxRA = 174$, and $\minDA = 436$. In cooperative baselines, agents perform one consensus update per time step. 
Before the stressor occurs $(t < 1400)$, all methods rapidly converge to the optimal action, achieving nearly identical cumulative total rewards. The differences become more evident in the cumulative regret curves, where LightCoop-Kripke exhibits the highest regret. This is primarily because each agent explores independently to resolve epistemic contradictions, rather than exploiting shared information as in the cooperative modes or committing more to exploitation as in Independent-UCB, resulting in a minor increase in cumulative regret. Looking at total reward over time plot in~\ref{fig:sigma05:total_reward}, we can see that the cooperative-UCB and cooperative-Kripke agents exhibit faster convergence in the very beginning, owing to shared statistics through consensus. Nevertheless, both Cooperative-UCB and Cooperative-Kripke exhibit minor transient dips in reward before the stressor. In Cooperative-Kripke, these deviations correspond to occasional false detections of regime change triggered by the residual-based gate. Because agents share information through consensus, false alarms can become temporally correlated, leading multiple agents to enter evidence-gathering simultaneously and momentarily reduce total reward. \addb{Although these pre-stressor deviations are visually noticeable, their magnitude and duration are limited. In Cooperative-Kripke, each false-alarm event produces an instantaneous reward drop of approximately $10-25\%$ and lasts for about $40-50$ time steps. As evidenced in the cumulative reward (Figure~\ref{fig:sigma05:cum_reward}) and cumulative regret (Figure~\ref{fig:sigma05:cum_regret}), these short dips have negligible long-term effect. That is, Cooperative-Kripke accumulates about the same pre-change total reward as the other methods and maintains the lowest cumulative regret. Thus, the transient dips reflect brief exploratory corrections rather than meaningful performance degradation.} In contrast, LightCoop-Kripke experiences no visible pre-change degradation since its detections occur asynchronously and do not synchronize the agents’ probing actions. A similar dip in Cooperative-UCB arises from consensus-induced bias: limited communication rounds cause agents’ local estimates to temporarily align toward a suboptimal arm, resulting in a short-lived collective exploration phase. Independent-UCB remains stable in this period because exploration errors across agents are uncorrelated and average out in aggregate performance. Immediately after the stressor at $t=1400$, all learning modes experience a dip in  total reward, fraction of agents selecting the optimal action, and higher levels of cumulative regret compared to levels before the occurrence of stressor. Specifically, both independent-UCB and cooperative-UCB baselines experience a pronounced degradation in all metrics as forgetting only is not sufficient to enforce timely re-exploration to identify the new optimal action. Nevertheless, Cooperative-UCB baseline experiences a deeper but shorter dip in reward and consequently higher cumulative regret and lower cumulative total reward compared to Independent-UCB. This occurs because, immediately after the change, the shared consensus variables cause all agents to temporarily rely on outdated discounted global estimates, which delays independent re-exploration. Once a few agents begin to update toward the new optimal action, the consensus mechanism rapidly propagates this information through the network, leading to a sharp collective recovery. In contrast, Independent-UCB agents respond asynchronously: each agent explores on its own and gradually discovers the new optimum at different times. This de-synchronization produces a smoother aggregate recovery with smaller instantaneous loss but slower overall adaptation.


In contrast, the Kripke-based learners demonstrate significantly faster epistemic and action recovery. The LightCoop-Kripke (independent learning with epistemic communication) quickly detects the regime shift and begins realigning, but requires additional time for all agents to stabilize on the new optimum since action learning remains independent. The fully Cooperative-Kripke mode, which integrates both epistemic and cooperative action layers, achieves the shortest recovery window and the smallest transient regret after the change. This mode rapidly rebuilds global consensus on the new hypothesis and synchronizes agent actions towards optimal actions. As shown in subfigures~\ref{fig:sigma05:cum_reward} and~\ref{fig:sigma05:cum_regret}, its cumulative reward curve closely tracks the pre-change trend with minimal slope reduction, and its cumulative regret saturates at the lowest level among all methods. Subfigure~\ref{fig:sigma05:fract_optimal} further confirms faster re-synchronization of the LightCoop-Kripke and Cooperative-Kripke agents, while independent and cooperative baselines exhibit slower transitions. Overall, the results validate that epistemic coordination significantly improves both recovery speed and stability after abrupt environmental shifts. Notably, even with limited communication (used only during epistemic recovery) the LightCoop-Kripke agents outperform baselines that rely on continuous communication for exchanging full statistics. It also performs slightly lower than the fully Cooperative-Kripke framework in which agents collaboratively recover while also sharing statistics. This means that  LightCoop-Kripke offers notable performance while being communication-efficient (see Table~\ref{table:scale}).

\textbf{High noise level}. Figure~\ref{fig:sigma10:main} summarizes the performance of all learning modes; Independent Discounted-UCB, Cooperative Discounted-UCB, and the proposed Independent-Kripke and Cooperative-Kripke frameworks when $\sigma = 1.0$. For this setup, the parameters for both Kripke learning mode is $\varepsilon_1 = 1.6$, $e_{\text{th}} = 13$, $L=30$, and $\etaE =10$, $\maxRE = 550$, $\minDE = 600$, $\maxRA = 174$, and $\minDA = 436$. At the higher noise level, the general performance ordering among the methods remains consistent with the moderate-noise case, but several important differences emerge. Notably, both Cooperative-Kripke and Cooperative-UCB no longer exhibit the small pre-stressor dips observed at lower variance. The increased noise effectively smooths out small residual fluctuations, reducing the likelihood of false detections and consensus-induced oscillations. After the stressor, the Independent-UCB shows markedly higher regret and lower cumulative reward than Cooperative-UCB, indicating that communication becomes more beneficial as uncertainty increases. In this regime, cooperation helps agents average out noisy reward observations and accelerates re-estimation of the new means, leading to a faster collective recovery. The Cooperative-Kripke again achieves the fastest and most stable recovery, while the LightCoop-Kripke maintains intermediate performance that is slightly below the fully cooperative mode but still outperforming both UCB baselines. This again suggest that LightCoop-Kripke offers comparable performance while achieving communication efficiency (see Table~\ref{table:scale}). Overall, these results suggest that communication and epistemic coordination provide greater advantage in high-variance environments, where shared information mitigates noise and enhances resilience to abrupt regime shifts.

\textbf{Effect of forgetting rate}: Figure~\ref{fig:sigma10ADD:main}  compares the Cooperative-Kripke framework with the Cooperative Discounted-UCB baselines under different forgetting rates. All other simulation parameters are identical to those used in Figure~\ref{fig:sigma10:main}. Across all metrics, Cooperative-Kripke consistently achieves the highest cumulative reward and the lowest cumulative regret, demonstrating both rapid and complete recovery after the change. In contrast, the performance of the cooperative UCB baselines depends strongly on the forgetting rate. Larger forgetting rates (e.g., $0.998$) lead to slower adaptation because past samples dominate the estimates, causing delayed response to the new reward means. Smaller forgetting rates (e.g., $0.990$) accelerate recovery by discounting old information more aggressively but at the cost of higher post-recovery variance and a lower steady-state reward level, as agents over-weight recent noisy observations. Thus, increasing the forgetting rate improves responsiveness but reduces asymptotic accuracy. The Cooperative-Kripke framework avoids this trade-off by detecting regime shifts explicitly and resetting beliefs only when epistemic evidence supports a change, enabling both fast recovery and convergence to the new optimum without sacrificing steady-state performance.

\begin{figure*}[t]
  \centering
  \subfloat[Epistemic Recovery Time vs $\etaE$.\label{fig:epi_rec}]{
    \includegraphics[width=0.48\textwidth]{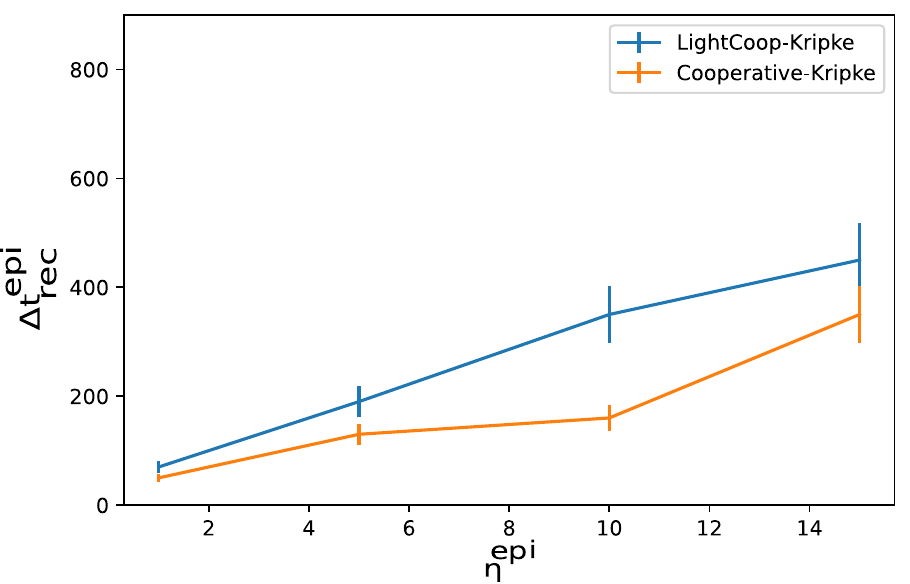}
  }\hfill
  \subfloat[Epistemic Durability Time vs $\etaE$.\label{fig:epi_dur}]{
    \includegraphics[width=0.48\textwidth]{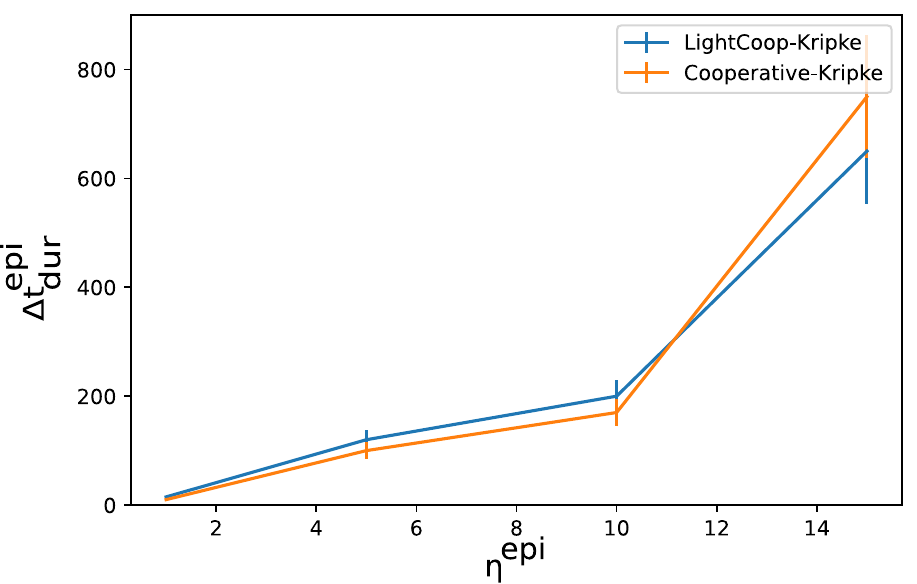}
  }
  \caption{Epistemic recovery and durability time vs different levels of $\etaE$ under noise level ($\sigma=1.0$). Each panel reports the mean across 10 trials with 95\% confidence intervals.}
  \label{fig:epi_tradeoff_main}
\end{figure*}

\textbf{Recoverability-Durability Tradeoff}: Figure~\ref{fig:epi_tradeoff_main} illustrates the intrinsic trade-off between epistemic recovery and durability as the evidence threshold $\etaE$ increases. A larger threshold imposes stricter evidence requirements before committing to a new world hypothesis, leading to a longer evidence accumulation phase and hence larger $\Delta\trecE$. At the same time, the higher threshold also enhances confidence in the committed hypothesis, yielding greater durability and longer $\Delta\tdurE$ after recovery. This relationship reflects the theoretical bounds introduced in Section~\ref{subsec:achieve_epi_rec}, where $\etaE$ effectively tunes the balance between the maximum admissible recovery time $\maxRE$ and the minimum required durability time $\minDE$. Across all thresholds, LightCoop-Kripke exhibits consistently longer recovery times than Cooperative-Kripke, as each agent must gather sufficient local evidence before reaching consensus on a change. The cooperative mode benefits from pooled information and thus attains faster epistemic alignment. Conversely, LightCoop-Kripke achieves slightly higher durability once recovered, since its individual decisions are less susceptible to occasional network-wide perturbations. At high $\etaE$, both methods remain stable until the end of the simulation horizon; consequently, the apparent reversal in durability ordering simply reflects that the method recovering earlier maintains its new belief for a longer remaining period. Overall, the results confirm the predicted recoverability–durability trade-off and demonstrate that cooperative Kripke provides a favorable operating point; achieving fast and stable epistemic adaptation.

Furthermore, in our experiments, the action recovery time $\Delta \trecA$ is few time steps after epistemic recovery (i.e., the first decision immediately following $\trecE$). This behavior arises because agents in both Kripke modes reinitialize their UCB learners using the committed epistemic states, causing the UCB indices of the optimal action to dominate immediately. We also observe that action durability closely tracks epistemic durability, differing only by a small offset corresponding to the short action recovery phase. Once the group’s epistemic states remain aligned with the correct world, the UCB bonuses shrink rapidly and the optimal action stays strictly preferred. Departures from the optimal action occur only if the epistemic layer re-enters evidence mode due to a false alarm or if a new stressor is introduced. In detail, for all our experiments, we instantiate $\Pi_i^{\mathrm{ext}}$ with the action set and realize best-policy identification via Gaussian-UCB. The confidence level $\etaA$ is implicit in the UCB exploration bonus and the reinitialization at $\trecE$; hence, the statistical separation condition is satisfied almost immediately.
the separation condition is satisfied almost immediately. Moreover, our implementation uses online bandit updates where the reinitialization step at commit is equivalent to performing a small number of virtual rollouts per action,

Across all regimes, epistemic coordination is the key driver of resilience: our proposed epistemic policy yields shorter recovery windows and higher post-recovery stability than fading-memory baselines. The cooperative-Kripke mode achieves the fastest recovery and lowest regret by pooling both epistemic knwoledge and action statistics. Meanwhile, LightCoop-Kripke (the event-triggered variant) still outperforms independent/cooperative-UCB baselines while using far less communication. It also performs very close to the Cooperative-Kripke while achieving communication efficiency. We also show that higher noise amplifies the value of cooperation and that  forgetting-rate tuning in discounted-UCB faces a speed–stability trade-off that the epistemic layer largely avoids by revising only when supported by evidence. 

\begin{table*}[h]
\centering
\caption{\addb{Performance comparison under Ring vs. Small-world topology. Total recovery time is
measured as the total resilience delay
$\Delta\trecE + \Delta\trecA$, i.e., the first time after the disturbance at
which the average reward returns to its steady-state maximum level (within a
small tolerance).}}
\label{table:scale}
\begin{tabular}{c c c c c c c c}
\toprule
\addb{\makecell{Number \\ of agents}} & \addb{Algorithm} & \multicolumn{3}{c}{\addb{Ring Topology}} & \multicolumn{3}{c}{\addb{Small-world Topology}} \\
\cmidrule(lr){3-5} \cmidrule(lr){6-8}
 & & \makecell{Total \\ recovery \\ time} & \makecell{Avg reward \\ after \\ recovery} & \makecell{Total \\ communication\\ (per-agent \\ per-step)} & \makecell{Total \\ recovery \\ time} & \makecell{Avg reward \\ after \\ recovery} & \makecell{Total \\ communication \\(per-agent \\ per step)} \\
\midrule
\multirow{5}{*}{$10$} 
 & \addb{Independent-UCB}              & $1100$ & $0.95$ &  $-$ & $1100$ & $0.95$ & $-$ \\
 & \addb{Cooperative-UCB}               & $600$  & $\boldsymbol{1.20}$ & \makecell{$50,000(2)$} & $500$  & $1.18$ & $100,000(4)$ \\
 & \addb{LightCoop-Kripke}      & $400$  & $1.19$ & $\boldsymbol{20 (8 \times 10^{-4})}$ & $360$  & $1.19$ & $\boldsymbol{24(9.6 \times 10^{-4})}$ \\
 & \addb{LightCoop-Kripke-fast} & $350$  & $1.19$ & $\boldsymbol{20(8 \times 10^{-4})}$ & $350$  & $1.19$ & $\boldsymbol{24(9.6 \times 10^{-4})}$ \\
 & \addb{Cooperative-Kripke}            & $\boldsymbol{160}$  & $\boldsymbol{1.22}$ & \makecell{$50,010(\sim 2)$} & $\boldsymbol{150}$  & $\boldsymbol{1.20}$ & $100,020(\sim 4)$ \\
\addlinespace[3pt]
\midrule[\heavyrulewidth]
\multirow{5}{*}{$150$} 
 & \addb{Independent-UCB}              & $1100$ & $0.95$ & $-$ & $1100$ & $0.95$& $-$ \\
 & \addb{Cooperative-UCB}              & $650$  & $\boldsymbol{1.20}$ & $750,000(30)$ & $500$  & $1.18$ & $1,500,000(60)$ \\
 & \addb{LightCoop-Kripke}     & $750$  & $1.00$ & $\boldsymbol{15,750(0.63)}$ & $600$  & $1.10$ & $\boldsymbol{4,860(0.19)}$ \\
 & \addb{LightCoop-Kripke-fast} & $400$  & $1.15$ & $\boldsymbol{15,750(0.63)}$ & $400$  & $1.18$ & $\boldsymbol{4,860(0.19)}$\\
 & \addb{Cooperative-Kripke}           & $\boldsymbol{150}$  & $\boldsymbol{1.20}$ & $765,750(30.6)$ & $\boldsymbol{120}$  & $\boldsymbol{1.20}$ & $1,505,100(60.2)$ \\
\addlinespace[3pt]
\midrule[\heavyrulewidth]
\multirow{5}{*}{$300$} 
 & \addb{Independent-UCB}              & $1100$ & $0.95$ & $-$ & $1100$ & $0.95$ & $-$\\
 & \addb{Cooperative-UCB}               & $650$  & $\boldsymbol{1.20}$ & $1,500,000(60)$ & $500$  & $\boldsymbol{1.18}$ & $3,000,000(120)$\\
 & \addb{LightCoop-Kripke}      & $950$  & $1.00$ & $\boldsymbol{54,000(2.16)}$& $800$  & $1.10$ & $\boldsymbol{13,086(0.52)}$\\
 & \addb{LightCoop-Kripke-fast} & $410$  & $1.15$ & $\boldsymbol{54,000(2.16)}$ & $400$  & $1.15$ &$\boldsymbol{13,086(0.52)}$ \\
 & \addb{Cooperative-Kripke}            & $\boldsymbol{150}$  & $\boldsymbol{1.20}$ & $1,581,000(63.24)$ & $\boldsymbol{130}$  & $1.17$ &$3,017,280(120.7)$ \\
\bottomrule
\end{tabular}
\end{table*}

\addb{\textbf{Scalability}: To evaluate scalability, we repeated the same setup used in Figure~\ref{fig:sigma10:main} under increasing network sizes
$|\agentSet|\in\{10,150,300\}$, while maintaining a ring communication
topology. Our analysis in Section~\ref{subsec:achieve_epi_rec} ( see Equation~\eqref{eq:boundEp}) indicates the recovery delay
$\Delta\trecE$ of LightCoop-Kripke increases with the diameter of the communication graph $\lcomm$. This means that in ring topology, the recovery time of LightCoop-Kripke $\Delta\trecE$ grows noticeably with the number of agents. This behavior is a direct consequence of the ring’s
hop diameter $\ell_{\mathrm{comm}}=\Theta(|\agentSet|)$: epistemic commits
must propagate sequentially around the network, slowing mutual knowledge
formation as the graph expands. In contrast, Cooperative-Kripke remains
largely unaffected by scaling, because its continuous consensus mixing
maintains synchronized beliefs and prevents long commit-propagation
phases.}

\addb{To mitigate the topology-induced delay, we consider two improvements.
First, we replace the ring with a small-world network, which preserves
sparsity but has a diameter that grows only logarithmically in
$|\agentSet|$. Such topologies present more realistic neighbor-connectivity patterns than a ring topology. Second, we test
a LightCoop-Kripke-fast variant that allows agents to commit immediately
once their local LLR surpasses $\etaE$, while broadcasting without waiting for
full horizon confirmation; if a higher-evidence hypothesis later arrives,
the agent revises accordingly. This reduces reliance on $\ell_{\mathrm{comm}}$
for timely epistemic alignment.}

\addb{Table~\ref{table:scale} examines scalability of the five methods in terms of total recovery time $\Delta\trecE +\Delta\trecA$, average reward after recovery (i.e., steady-state maximum level within a small tolerance), and total communication cost (with normalized per-agent per-step cost in parentheses). Total communication cost measures the overall number of messages transmitted over graph edges across all agents and all rounds, including both consensus-updates (if any) and broadcast floods used for epistemic commits (if any). This captures the full communication overhead of the algorithm. Independent-UCB incurs no communication overhead but consistently exhibits the slowest recovery and lowest steady-state reward, indicating poor resilience in dynamic environments. Cooperative-UCB achieves fast recovery and high average reward after recovery, with recovery time stays flat as the network grows; however, this comes at the expense of very high communication overhead, which increases further in small-world topologies due to higher node degree. That is, in our setup, each agent in the ring has degree $2$ (two neighbors), whereas in the small-world graph it has degree $4$ on average. Because consensus messages are exchanged with every neighbor at each step, this doubling of node degree roughly doubles the per-step communication load, which explains why the total communication cost in the small-world topology is about twice that of the ring. Cooperative-Kripke requires slightly more communication than Cooperative-UCB because of epistemic evidence sharing but achieves substantially faster recovery that stays flat as well with growing number of agents and high average reward after recovery. When communication efficiency is a priority, LightCoop-Kripke offers a favorable trade-off, achieving significant reductions in communication cost even at $\agentsize = 300$, but its total recovery time grows with number of agents under ring topology (since $\lcomm$ grows as well. LightCoop-Kripke-fast eliminates the diameter-waiting delay by allowing agents to commit immediately once sufficient local evidence is gathered, maintaining low communication cost while maintaining recovery time nearly flat with growing number of agents and achieving near high average reward after recovery. Note that LightCoop-Kripke and LightCoop-Kripke-fast incur essentially the same communication cost because they use the same Kripke evidence layer: agents only transmit hypothesis-broadcast messages when their local log-likelihood margin crosses the epistemic threshold, and each such broadcast is forwarded once along each edge. The ``fast'' variant only changes when agents decide to commit (they no longer wait $\lcomm$ steps before updating their world model), but it does not change how many broadcast messages are sent, so the total communication load remains effectively identical. Because small-world topologies dramatically reduce the diameter of the network, information about new hypotheses propagates to all agents much more rapidly. As a result, agents are less likely to initiate redundant broadcasts since they learn of the change from others faster resulting in substantially reduced communication overhead compared to ring topologies. These results demonstrate that the observed slowdown of LightCoop-Kripke on a ring is not a fundamental limitation of epistemic coordination, but a topology-induced effect: under realistic small-world connectivity, both recovery time and communication cost remain well-controlled even at $\agentsize = 300$, with the FAST variant delivering scalable resilience with minimal communication overhead.} 



\section{Conclusion and Future Work}
\label{sec:conclusion}
This work presented a novel formal framework for quantifying and designing multi-agent resilience in dynamic systems. We introduced a Kripke-based learning architecture that integrates epistemic resilience (the ability to detect, share, and recover accurate knowledge of world conditions) with action resilience (the ability to realign and sustain optimal behavior after disruptions). Our formulation links these dimensions through measurable recovery and durability times, providing a principled definition of resilience that bridges learning, communication, and control. Through extensive simulations, we demonstrated that the proposed Cooperative-Kripke and LightCoop-Kripke architectures outperform classical baselines such as independent and cooperative discounted-UCB. The results show that epistemic coordination significantly improves both recovery speed and post-change stability, even under limited communication. In particular, Cooperative-Kripke achieves rapid epistemic and action recovery while maintaining high cumulative rewards and minimal regret, whereas LightCoop-Kripke achieves slightly lower performance but with substantially reduced communication load. The experiments also verified the recoverability–durability trade-off, showing that stricter evidence thresholds prolong the evidence collection phase albeit longer stability after recovery.

A promising research avenue lies in formulating the resilience design problem as a Mixed-Integer Linear Program (MILP). Such a formulation can encode the logical constraints of epistemic transitions, the temporal budgets on recovery and durability ($\maxRE, \minDE, \maxRA, \minDA$), communication and computation budgets as linear and integer constraints, enabling exact computation of feasible resilience specifications or optimal parameter selections ($\etaE,\etaA$) under communication and sensing limits. Solving this MILP provides a systematic method for synthesizing resilient agent configurations and verifying resilience guarantees before deployment. This optimization perspective complements the learning-based approach by offering design-time verification and run-time resilience in a unified framework.

\addb{
An additional direction for future work is the incorporation of explicit
collision-aware dynamics. In the present formulation, interactions among agents
are treated as part of the environment and therefore enter each agent’s
epistemic state. This allows potential collisions to be interpreted as
unexpected observations that trigger the same epistemic-update mechanisms used
for other stressors. However, extending the framework with explicit
collision models (such as contention-aware reward structures, joint-action
feasibility constraints, or coordination protocols) would enable proactive
avoidance and richer multi-agent interactions in collision-prone settings.
}\addb{Finally, while our algorithms incorporate safeguards against noisy observations through residual thresholds, exceedance counting, and evidence margins, modeling unreliable communication remains an important future extension.}


\bibliographystyle{IEEEtran}
\bibliography{references}

\end{document}